\documentclass[11pt,letter,reqno]{amsart}
\usepackage[numbers]{natbib}
\usepackage{graphicx}
\usepackage{comment}
\usepackage{mathrsfs}
\usepackage{mathabx}
\usepackage{MnSymbol}
\usepackage{setspace}
\usepackage{amsaddr,amsthm}
\usepackage{thmtools,thm-restate}
\usepackage{multirow}
\usepackage{arydshln}
\usepackage{chngcntr}
\usepackage{algorithm,caption,algorithmic}
\usepackage[toc,page]{appendix}
\usepackage{etoolbox}

\appto\appendix{\addtocontents{toc}{\protect\setcounter{tocdepth}{0}}}

\makeatletter
\renewcommand{\tocsection}[3]{%
	\indentlabel{\@ifnotempty{#2}{\bfseries\ignorespaces#1 #2\quad}}\bfseries#3}
\renewcommand{\tocsubsection}[3]{%
	\indentlabel{\@ifnotempty{#2}{\ignorespaces#1 #2\quad}}#3}
\newcommand\@dotsep{4.5}
\def\@tocline#1#2#3#4#5#6#7{\relax
	\ifnum #1>\c@tocdepth 
	\else
	\par \addpenalty\@secpenalty\addvspace{#2}%
	\begingroup \hyphenpenalty\@M
	\@ifempty{#4}{%
		\@tempdima\csname r@tocindent\number#1\endcsname\relax
	}{%
		\@tempdima#4\relax
	}%
	\parindent\z@ \leftskip#3\relax \advance\leftskip\@tempdima\relax
	\rightskip\@pnumwidth plus1em \parfillskip-\@pnumwidth
	#5\leavevmode\hskip-\@tempdima{#6}\nobreak
	\leaders\hbox{$\m@th\mkern \@dotsep mu\hbox{.}\mkern \@dotsep mu$}\hfill
	\nobreak
	\hbox to\@pnumwidth{\@tocpagenum{\ifnum#1=1\bfseries\fi#7}}\par
	\nobreak
	\endgroup
	\fi}
\AtBeginDocument{%
	\expandafter\renewcommand\csname r@tocindent0\endcsname{0pt}
}
\def\l@subsection{\@tocline{2}{0pt}{2.5pc}{5pc}{}}
\makeatother

\makeatletter
\renewcommand\section{\@startsection {section}{1}{\z@}%
	{-3.5ex \@plus -1ex \@minus -.2ex}%
	{2.3ex \@plus.2ex}%
	{\normalfont\@secnumfont\fontsize{15}{18}\bfseries}}
\renewcommand\subsection{\@startsection{subsection}{2}{\z@}%
	{-3.25ex\@plus -1ex \@minus -.2ex}%
	{1.5ex \@plus .2ex}%
	{\normalfont\fontsize{13.5}{17}\selectfont}}
\renewcommand\subsubsection{\@startsection{subsubsection}{3}{\z@}%
	{-3.25ex\@plus -1ex \@minus -.2ex}%
	{1.5ex \@plus .2ex}%
	{\normalfont\normalsize\fontsize{12.5}{17}\selectfont}}

\def\@seccntformat#1{%
	\protect\textup{\protect\@secnumfont
		\ifnum\pdfstrcmp{section}{#1}=0 \bfseries\fi
		\csname the#1\endcsname
		\protect\@secnumpunct
	}%
}  
\makeatother

\calclayout
\allowdisplaybreaks

\numberwithin{equation}{section}

\theoremstyle{plain}
 
\newtheorem{defin}{Definition}

\newcommand{\bc}{\boldsymbol}
\newcommand{\mc}{\mathcal}

\makeatletter
\def\blfootnote{\gdef\@thefnmark{}\@footnotetext}
\makeatother

\author{Taeho Kim$^\star$ and Edsel A. Pe$\tilde{\text{n}}$a$^\dagger$}

\title{\Large{Improved Multiple Confidence Intervals via \\Thresholding Informed by Prior Information}}
\address{\begin{tabular}{cc}Department of Statistics&Department of Statistics\\University of Haifa& University of South Carolina\\Haifa, 31905. Israel& Columbia, SC 29208. U.S.A.\end{tabular}}

\email{ktaeho@campus.haifa.ac.il; pena@stat.sc.edu}
\let\origmaketitle\maketitle
\def\maketitle{
	\begingroup
	\def\uppercasenonmath##1{} 
	\let\MakeUppercase\relax 
	\origmaketitle
	\endgroup
}

\begin{document}

\maketitle

\blfootnote{\textup{$^\star$T. Kim is a Postdoctoral Researcher,} Department of Statistics, University of Haifa. Israel.}
\blfootnote{\textup{$^\dagger$E. A. Pe$\tilde{\text{n}}$a is a Professor,} Department of Statistics, University of South Carolina, Columbia, SC. USA.}

\begin{abstract}
Consider a statistical problem where a set of parameters are of interest to a researcher. 
Then multiple confidence intervals can be constructed to infer the set of parameters simultaneously. 
The constructed multiple confidence intervals are the realization of a multiple interval estimator (MIE), the main focus of this study. 
In particular, a thresholding approach is introduced to improve the performance of the MIE. 
The developed thresholds require additional information, so a prior distribution is assumed for this purpose. 
The MIE procedure is then evaluated by two performance measures: a global coverage probability and a global expected content, which are averages with respect to the prior distribution. 
The procedure defined by the performance measures will be called a Bayes MIE with thresholding (BMIE Thres). In this study, a normal-normal model is utilized to build up the BMIE Thres for a set of location parameters. Then, analytic behaviors of the BMIE Thres are investigated in terms of the performance measures, which approach those of the corresponding $z$-based MIE as the thresholding parameter, $C$, goes to infinity. In addition, an optimization procedure is introduced to achieve the best thresholding parameter $C$. For illustration purposes, in-season baseball batting average data and leukemia gene expression data are used to demonstrate the procedure for the known and unknown standard deviations situations, respectively. In the ensuing simulations, the target parameters are generated from different true generating distributions in order to consider the misspecified prior situation. The simulation also involves (empirical) Bayes credible MIE, and the effectiveness among the different MIEs are compared with respect to the performance measures. In general, the thresholding procedure helps to achieve a meaningful reduction in the global expected content while maintaining a nominal level of the global coverage probability.\\
\vspace{3pt}

\vspace{2pt}


\end{abstract}

\newpage

\tableofcontents


\section{Introduction}
Suppose an interval estimator (IE) is constructed for a single parameter, $\theta_0$. With a given level of $1-\alpha$, the IE is essentially a set-valued measurable mapping, $\Gamma(\cdot;\theta_0,\alpha)$, from a sample space $\mc{X}$ to the sigma field of a parameter space $\Theta$. Note that a family of probability distributions, $\mc{P}=\{P_\theta:\theta\in\Theta\}$, on the sample space is postulated so that it forms a statistical model. Once the IE is constructed, then the precision and accuracy of the estimation can be evaluated through expected length (EL) and coverage probability (CP): 
\begin{align}\label{eq11}
EL[\theta_0,\alpha]=E_{\theta_0}[\nu(\Gamma(X;\theta_0,\alpha))]\;\&\;CP[\theta_0,\alpha]=P_{\theta_0}[\theta_0\in\Gamma(X;\theta_0,\alpha)]
\end{align}
where $\nu$ measures the length of the IE. Based on these performance measures, the usual optimality condition for an IE is to minimize EL while maintaining CP at a given level of at least $1-\alpha$.

Inverting a test function is a common approach to constructing an IE. This exploits the concept of duality, the following correspondence between interval estimation and hypothesis testing. Given $\alpha\in(0,1)$, a test function of $\alpha$ size for $H_0:\;\theta=\theta_0$ vs. $H_A:\;\theta\ne\theta_0$ is a mapping $\delta(\cdot;\theta_0,\alpha):\;\mathcal{X}\to\{0,1\}$ such that $P_{\theta_0}[\delta(X;\theta_0,\alpha)=1]\leq\alpha$. Then, the corresponding IE is 
\begin{align}\label{eq12}
\Gamma(x;\theta_0,\alpha)=\{\theta\in\Theta:\;\delta(x;\theta_0,\alpha)=0\}.
\end{align}
Conversely, suppose $\Gamma(\cdot;\theta_0,\alpha)$ is an interval estimator of $1-\alpha$ level for $\theta_0$ such that $P_{\theta_0}[\theta_0\in\Gamma(X;\theta_0,\alpha)]\geq 1-\alpha$. Then the corresponding test function for $H_0:\;\theta=\theta_0$ vs. $H_A:\;\theta\ne\theta_0$ is
\begin{align}\label{eq13}
	\delta(x;\theta_0,\alpha)=I\{\theta\notin \Gamma(x;\theta_0,\alpha)\}.
\end{align}
From the constructions, it is clear that the CP of the resulting IE in (\ref{eq12}) is at least $1-\alpha$ and the type-I error rate of the resulting test function in (\ref{eq13}) is at most $\alpha$, respectively.
A similar logic holds true for the case of multiple parameters: $\bc{\theta_0}=\left(\theta_1^0,\theta_2^0,\ldots,\theta_M^0\right)$. That is, a duality exists between a multiple testing procedure (MTP) and a multiple interval estimator (MIE), so that the structure of MTP can be transferred to that of corresponding MIE, and vice versa. From the perspective of developing MIEs, this is particularly useful to derive a global coverage probability from an existing global type-I error rate which has been extensively studied in the field of multiple testing.

The probability of committing at least one type-I error is called the family-wise error rate (FWER), and it is one of the most well-established global type-I error rates (\citet{hochberg:1987}). Suppose $(\bc{\mc{X}},\bc{\mc{F}},\bc{\mc{P}})=\left(\bigotimes_{m=1}^M\mc{X}_m,\sigma\left(\bigotimes_{m=1}^M\mathcal{F}_m\right),\prod_{m=1}^M\mc{P}_m\right)$ is a statistical model and $(\bc{\Theta,\mc{T}})=\left(\bigotimes_{m=1}^M\Theta_m,\sigma\left(\bigotimes_{m=1}^M\mc{T}_m\right) \right)$ is a product-measurable space. In this study, we assume independence of random quantities throughout the index $m$. Then an MTP of global size $q$ for $H^0_m:\;\theta_m=\theta_m^0$ vs. $H^A_m:\;\theta_m\ne\theta_m^0$ for $m=1,2,\ldots,M$ is a mapping $\bc{\delta(x;\theta_0,\alpha)}=\left(\delta_m(x_m;\theta_m^0,\alpha_m)\right)_{m=1}^M$ such that 
\begin{align}\label{eq14}
\text{FWER}\left[\bc{\theta_0,\alpha}\right]:=\bc{P_{\theta_0}}\left[\bigcup_{m=1}^M\left\{\delta_m\left(X_m;\theta_m^0,\alpha_m\right)=1\right\}\right]\leq q.
\end{align}
Note that this is the case where an FWER is weakly controlled; whereas, if the inequality holds with the probability evaluated by any possible combinations of the null and alternative hypotheses, then an FWER is strongly controlled. Now we utilize the FWER to derive a particular global coverage probability, called the family-wise coverage rate (FWCR). By exploiting the duality, an MIE is constructed as follows: 
\begin{align}\label{eq15}
\bc{\Gamma(x;\theta_0,\alpha)}=\bigtimes_{m=1}^M\Gamma_m\left(x_m;\theta_m^0,\alpha_m\right)=\bigtimes_{m=1}^M\left\{\theta_m^0\in\Theta_m:\delta_m\left(x_m;\theta_m^0,\alpha_m\right)=0\right\}.
\end{align}
Then FWCR is defined as follows:
\begin{align}\label{eq16}
\text{FWCR}[\bc{\theta_0,\alpha}]:=&\bc{P_{\theta_0}}\left[\bigcap_{m=1}^M\left\{\theta_m^0\in\Gamma_m\left(X_m;\theta_m^0,\alpha_m\right)\right\}\right]  \\
=&1-\bc{P_{\theta_0}}\left[\bigcup_{m=1}^M\left\{\theta_m^0\notin\Gamma_m\left(X_m;\theta_m^0,\alpha_m\right)\right\}\right] \nonumber \\
=&1-\bc{P_{\theta_0}}\left[\bigcup_{m=1}^M\left\{\delta_m\left(X_m;\theta_m^0,\alpha_m\right)=1\right\}\right]=1-\text{FWER}[\bc{\theta_0,\alpha}]\geq 1-q \nonumber
\end{align}
Note that the duality ensures the MIE maintains the FWCR at a global level of at least $1-q$. Given this condition, we seek an MIE which minimizes the global expected content represented by the average of the expected lengths. 
 
Now, to motivate a thresholding approach, let us focus on an MIE with two-sided individual IEs. If it is possible to determine that a subset of the target parameters resides on one side of the corresponding IEs, then the other side of the IEs can be removed in order to minimize the global expected content of the MIE. In order to implement the removal process, the MIE will be equipped with a pair of thresholds. However, setting up the thresholds requires additional information such as knowledge from the domain science, results from previous experiments, and/or common sense. In this study, we adopt prior information regarding the target parameters to fulfill the requirement. This will be modeled by setting a prior distribution, $\bc{\Pi}=\prod_{m=1}^M\Pi_m$, on the parameter space $\bc{\Theta}$ equipped with a sigma-field $\bc{\mathcal{T}}$. The location of the thresholds will be determined according to the given prior structure along with a thresholding parameter, $C$.

In general, the expected length and coverage probability of an IE depend on the true parameter value. The issue is that we never know the \textit{true} parameter. Moreover, it is quite rare to obtain an IE which uniformly dominates other IEs. Therefore, we usually summarize the performance measures throughout the set of parameter values. For example, a confidence coefficient, the infimum of coverage probabilities, is the essential summarization used to define a classical confidence interval. However, this is not the only method to summarize the performance. In particular, note that our study assumes the existence of prior information to implement the thresholding approach; thus, we would involve this information by integrating the performance measures with respect to the prior distribution. Because this particular integration is reminiscent of the derivation of the Bayes risk, the resulting measures are denoted as the Bayes expected length (BEL) and the Bayes coverage probability (BCP):
\begin{align}
BEL[\theta,\alpha]&=\int_{\theta\in\Theta}E_{\theta}[\nu(\Gamma(X;\theta))]d\Pi(\theta);\label{eq17}\\
BCP[\theta,\alpha]&=\int_{\theta\in\Theta}P_{\theta}[\theta\in\Gamma(X;\theta)]d\Pi(\theta).\label{eq18}
\end{align}
Moreover, the correponding individual procedure is called the Bayes IE (BIE), so that given $\alpha\in(0,1)$, suppose $X|\theta=(X_1,X_2,\ldots,X_m)|\theta\sim P_\theta$ and $\theta\sim\Pi$. Then, $\Gamma(\cdot;\theta,\alpha)$ is called a $100(1-\alpha)\%$ BIE for $\theta$ if $BCP[\theta,\alpha]\geq 1-\alpha$.
Notice that every classical confidence interval satisfies the condition of the BIE; however, the converse is not necessarily true. This is because the coverage probability at a specific parameter value may be lower than a nominal level, although the procedure maintains the BCP at least at the given level. The modification through the integration allows us to reflect the prior information to the EL and CP in (\ref{eq11}) without losing their characteristics. Therefore, we seek an optimal BIE which minimizes the BEL given a level $1-\alpha$. In a later section, this concept of Bayes interval estimation will be extended to the case with multiple parameters, and the procedure will be called a Bayes multiple interval estimator (BMIE). 

We review related studies in section 2. The BMIE with thresholding (BMIE Thres) for the location parameters of the normal-normal model is introduced in section 3; moreover, its analytic properties and behaviors are studied with respect to the thresholding parameter $C$. In section 4, a decision-theoretic optimization procedure is presented, and the optimal thresholding parameter $C^*$ is investigated. In section 5, the performance of the BMIE Thres is demonstrated by data applications under the known and unknown standard deviations. In-season baseball batting average data is applied to the procedure for the former case, and leukemia gene expression data is applied for the latter case. In section 6, we perform a simulation study to address the case of misspecified prior distributions, as well as to compare BMIE Thres with Bayesian credible MIEs. Lastly, discussion and suggestions for future work are presented in section 7.

\section{Related Studies}
The multiplicity issue is a fundamental problem whenever an inferential procedure attempts to handle a set of parameters simultaneously. Due to this issue, we cannot simply assign the usual 0.05 individual sizes to MTPs or 0.95 individual levels to MIEs because the global type-I error increases or the global coverage probability decreases as the number of parameters increases. (\citet{lehmann:2006}) In earlier studies, such as \citet{miller:1966}, one of the first published books on multiple inferences, the meaning of \textit{multiple} was usually about less than 10 parameters. However, in the 1990's, the dimensions of problems became much higher due to the influence of high-throughput data and the multiplicity issue was magnified in earnest. (\citet{efron:2012}) As a result, the amount of research in MTPs was boosted.   

The FWER is one of the classical global type-I error rate to incorporate the multiplicity issue in MTPs. Suppose we want to build up an MTP which controls the FWER for $M$ target parameters with a global size of $q$. An intuitive multiple adjustment for the individual sizes would be the Bonferroni approach, $q/M$. If we can assume the independence among the individual procedures, then the Sidak approach, $1-(1-q)^{1/M}$, would also be a valid adjustment. This means, by applying these individual sizes to the MTP, the FWER can be bounded above by the global size of $q$. However, these one-step approaches generally result in low global power, e.g., a limited number of rejections. (\citet{shaffer:1995}) In order to overcome this limitation, step-wise approaches were introduced by \citet{holm:1979} and \citet{hochberg:1988}. These procedures utilize the information of ordered $p$-values to assign particular sizes to the corresponding individual testings. Also, \citet{westfall:1993} suggested a resampling procedure in order to use the dependence structure of $p$-values to increase the global power. Another well-known approach is a $p$-value weighting. This approach seeks optimal weights for $p$-values to achieve a higher global power. Naturally, the issue is how to choose the optimal weights. To handle this, \citet{westfall:1998} and \citet{dobriban:2015} assumed prior information. Although the procedures utilize prior information to choose the weights, they never claimed their approaches to be Bayesian, as the procedures are aimed to control the FWER, a frequentist global type-I error rate. Instead, Dobriban called their approach quasi-Bayesian. \citet{pena:2011} considered the problem of multiple testing as a general problem under a decision theoretic framework. Their MTP allocates optimal sizes to individual tests to maximize a global power under the FWER and false discovery rate (FDR) by \citet{benjamini:1995}.  

In early works on multiple interval estimation, e.g., \citet{scheffe:1953} and \citet{roy:1953}, the researchers mainly considered the research field to be a part of multiple comparisons among a small number of parameters within the setting of ANOVA or regression. In addition, \citet{benjamini:2005} pointed out that the cases of ignoring the multiplicity adjustment in multiple interval estimations were more frequent than the cases in multiple testings even after the influence of high-throughput data in the 1990's. Furthermore, some existing studies introduced MIEs as supplementary procedures for established MTPs. However, without information from alternative hypotheses in MIEs, there exists no explicit relation between the global power in the MTPs and the global expected length in the MIEs. This implies that the multiple interval estimation is an independent topic and should be investigated independently. A good example is the relation between the FDR and the false coverage rate (FCR) by \citet{benjamini:2005}. Although these two concepts represent the global type-I error rate and global coverage probability as developed by the same group of researchers, the lack of alternative information forced the authors to introduce the concept of parameter selection in relation to another topic, selective inference. (\citet{fithian:2014})  

There are several studies that investigated MIEs using the empirical Bayes framework, and these are closely related to our study in terms of the modeling perspective. \citet{morris:1983} investigated the empirical Bayes interval estimation under the same setting as \citet{efron:1975}, which studied the point estimation. \citet{casella:1983} studied parametric empirical Bayes confidence sets for multivariate normal means. Due to the shrinkage effect, the empirical Bayes confidence sets provides shifted estimates, which result in better coverage probability. In the next section, we exploit the model setup of the parametric empirical Bayes for our MIE to establish a pair of thresholds. A similar thresholding idea was considered in \citet{habiger:2014} for MTPs to maximize a global power.

\section{Bayes Multiple Interval Estimator with Thresholding}
\subsection{Bayes Multiple Interval Estimator}
\leavevmode
Consider a statistical model $(\bc{\mathcal{X}},\bc{\mathcal{F}},\bc{\mathcal{P}})$ where $\bc{\mathcal{X}}=\bigotimes_{m=1}^M\mathcal{X}_m$ is a product sample space, $\bc{\mathcal{F}}=\sigma\left(\bigotimes_{m=1}^M\mathcal{F}_m\right)$ is an associated product $\sigma$-field, and $\bc{\mathcal{P}}=\prod_{m=1}^M\mc{P}_m$ is a class of probability distributions on the sample space. 
In addition, $(\bc{\Theta},\bc{\mathcal{T}},\bc{\Pi})$ is another probability space with $ \bc{\Theta}=\bigotimes_{m=1}^M\Theta_m$, $\bc{\mathcal{T}}=\sigma\left(\bigotimes_{m=1}^M\mathcal{T}_m\right)$, and $\bc{\Pi}=\prod_{m=1}^M\Pi_m$ is a class of prior probability distributions on the parameter space.
An $M$ dimensional random quantity $\bc{X}$ is generated from $\bc{P_\theta}$, where $\bc{X}$ consists of $X_m$s which are independent throughout the index $m=1,2,\ldots,M$. For each $\theta_m\in\Theta_m$, denote the prior density of $\Pi_m$ by $\pi_m$ and denote the density of $P_{\theta_m}$ by $f_{\theta_m}$; moreover, assume the mapping $(\theta_m,x_m)\mapsto f_{\theta_m}(x_m)$ is product-measurable. 

Given this setting, the coverage probability for an individual IE is summarized by the BCP in (\ref{eq18}), the coverage probability integrated with respect to the prior distribution. Note we defined the BIE based on the BCP. Now, we define the multiple extension of BIE as follows: 
\begin{defin}
	\label{defin2}
	Given $q\in(0,1)$, $100(1-q)\%$ Bayes Multiple Interval Estimator (BMIE) for $\bc{\theta}$ is a mapping, $\bc{\Gamma(\cdot)}: \bc{\mathcal{X}}\longrightarrow \bc{\mathcal{T}}$, such that 
	\begin{align} \label{eq31}
	\prod_{m=1}^M\int_{\Theta_m} P_{\theta_m}[\theta_m\in\Gamma_m(X_m;\theta_m,\alpha_m)]d\Pi_m(\theta_m)\geq 1-q.
	\end{align}
	The left-hand side quantity is the global coverage probablity of $\bc{\Gamma}$ and is called the Bayes family-wise coverage rate (BFWCR). 
\end{defin}
Note that we could involve another random quantity $\bc{U}$, which is independent of $\bc{X}$, from $M$ standard uniform distributions to consider as randomizers. With the additional random quantity, it is possible to design a randomized BMIE which has its BFWCR always equals to the global level, $1-q$. In this study, we omit this additional procedure for simplicity. Similarly, the expected length for an IE is summarized by BEL in (\ref{eq17}), the expected length integrated with respect to the prior distribution. With multiple parameters, the extended length of an BMIE is defined by a global expected content as follows:
\begin{defin}
	\label{defin3}
	The Bayes average expected length (BAEL) is the global expected content of a Bayes Multiple Intervel Estimator (BMIE) and is defined as follows:
	\begin{align} \label{eq32}
	\frac{1}{M}\sum_{m=1}^M	\int_{\Theta_m} E_{\theta_m}[\nu(\Gamma_m(X_m;\theta_m,\alpha_m))]d\Pi_m(\theta_m),
	\end{align}
	where $\nu$ is the content of individual IEs.
\end{defin}
Notice that the content, $\nu$, is a measure which quantifies
\textit{the general length} of IEs. It is general since $
\nu$ is not always the Lebesgue measure, which agrees with our usual perception of length. Instead, it can be chosen from a class of measures in accordance with the type of target parameter so that the Invariance Principle is satisfied.  
For example, whereas the content for a location parameter requires the Lebesgue measure $d\lambda(\theta)$, the content for a scale parameter requires another measure, $\tfrac{1}{\theta}d\lambda(\theta)$. Refer to chapter 6 of \citet{berger:2013} for details.

\subsection{Thresholding Idea for Bayes Multiple Interval Estimator}
\leavevmode
In this subsection, we introduce a particular BMIE equipped with a pair of thresholds to reduce its BAEL. This procedure is called an BMIE with thresholding (BMIE Thres). Suppose we are interested in $M$ normal means with known standard deviations. 
Random samples are observed from the normal distributions and these are independent throughout the index $m$ from 1 to $M$. In addition, we assume the prior information of the location parameters is available in the form of a normal distribution with the hyper-parameters, $\eta$ and $\tau$:
\begin{align}\label{setting1}
\bar{X}_m|\mu_m\sim \mathcal{N}\left(\mu_m,\sigma^2_m\right)\;\;\&\;\mu_m\sim \mathcal{N}\left(\eta,\tau^2\right)\;\text{for}\;m=1,2,\ldots,M.
\end{align}
Here, the problem is simplified by the Sufficiency Principle and let $\sigma_m$'s, which depend on the sample sizes, denote the standard errors, without loss of generality. Since we assume a common prior distribution for every location parameter, the setting becomes identical to the parametric empirical Bayes framework in \citet{efron:1975} and \citet{casella:1983}.
Recall the well known fact that the posterior distribution of the normal-normal model also follows a normal distribution (\citet{berger:2013}):
\begin{align}\label{setting2}
\mu_m|\bar{x}_m\stackrel{\text{ind.}}{\sim}\mathcal{N}\left(\hat{\mu}_m,\beta_m\sigma_m^2   \right)\;\text{for}\;m=1,2,\ldots,M
\end{align}
where $\hat{\mu}_m=\beta_m \bar{x}_m+(1-\beta_m)\eta$ and $\beta_m=\frac{\tau^2}{\tau^2+\sigma^2_m}$. In particular, the posterior mean $\hat{\mu}_m$ is a convex combination of the maximum likelihood estimate $\bar{x}_m$ and the prior mean $\eta$. 

\begin{figure}[ht]
	\centering 
	\includegraphics[scale=0.42]{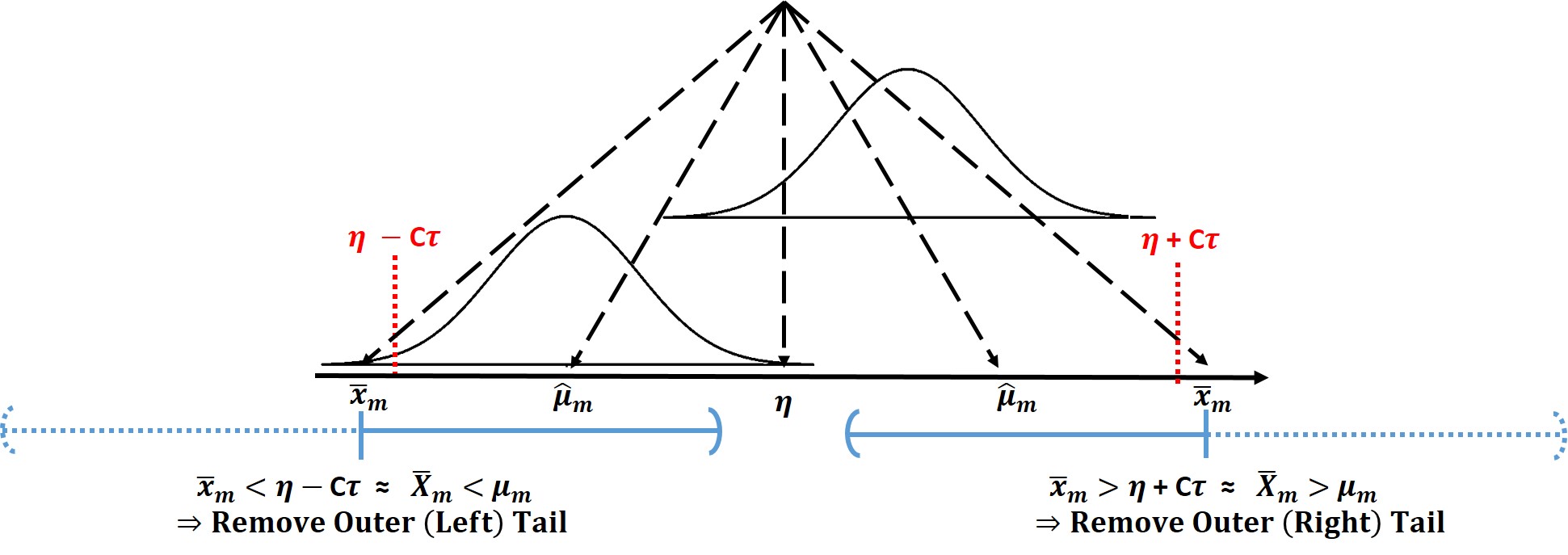} 
	\caption[Idea of Thresholding Approach]{Idea of Thresholding Approach} 
	\label{idea} 
\end{figure}
Now, to illustrate the idea of thresholding, suppose $\sigma_m$'s are identical and $\tau$ is fixed so that all the posterior distributions have a constant dispersion.
Next, envision a situation in which an observed $\bar{x}_m$ deviates from the prior mean, $\eta$. Then the distance between $\bar{x}_m$ and $\hat{\mu}_m$ becomes larger in proportion to the increase of the distance between $\bar{x}_m$ and $\eta$. Still, we construct an IE in a classical way, letting $\bar{x}_m$ to be the center of the $m$th observed IE. Let us call the inner tail for the side of IE close to the $\eta$ and the outer tail for the one in the opposite side. Note that as the estimate $\bar{x}_m$ differs more from the $\eta$, the posterior probability that $\mu_m$ resides in the inner tail becomes greater. In this situation, we would remove the outer tail to reduce the global expected content. 

The decision on the removal can be made by using a pair of thresholds on both sides: $\eta-C\tau$ and $\eta+C\tau$. Therefore, once an estimate $\bar{x}_m$ falls outside of the thresholds, then we discard the outer tail and keep only the inner tail.   
In actual situations, $\sigma_m$ would vary; however, the idea remains the same although the effect of thresholding is affected by the size of $\sigma_m$ in relation to the size of $\tau$.
With this motivation, the $m$th BIE with thresholding (BIE Thres) has the following form:
\begin{align}\label{eq33}
\Gamma_m(\bc{X}_m;\mu_m,\alpha_m)=\left(\bar{X}_m-z_{\alpha_m/2}\sigma_mI\left\{\bar{X}_m>\eta-C\tau\right\},\bar{X}_m+z_{\alpha_m/2}\sigma_mI\left\{\bar{X}_m<\eta+C\tau\right\}  \right)
\end{align}
where $z_{\alpha_m/2}=\Phi^{-1}\left(1-\alpha_m/2\right)$. Note that $\eta$ and $\tau$ will be estimated in the actual implementation. 
Therefore, although a single estimator is initially simple as it depends only on $\bar{X}_m$ with respect to the given level, $1-\alpha_m$, the BIE Thres becomes compound for $X_1,X_2,\ldots,X_M$ after the hyper-parameter estimation and optimization as in the later sections.   

\subsection{Individual Performance Measures}
\leavevmode
In this subsection, we ascertain the forms of performance measures of a single BIE Thres under the normal-normal model. In addition, their properties with respect to $C$ are investigated in relation to the classical $z$-based confidence interval estimator for a normal location parameter ($z$-based IE). First, the following proposition shows the form and property of the BEL of the $m$th BIE Thres:
\begin{restatable}{prop}{propa}
	\label{prop1}
	The Bayes expected length (BEL) of the $m$th BIE Thres has the form of 
	\begin{align} \label{eq34}
	BEL[\mu_m,\alpha_m,C;\sigma_m,\tau]=2z_{\alpha_m/2}\sigma_m\Phi(C_m)
	\end{align} 
	where $C_m=C\tau/\sqrt{\sigma_m^2+\tau^2}$, and the BEL approaches the expected length of the corresponding $z$-based IE as $C$ goes to $\infty$.
\end{restatable}
Notice that the BEL consists of two parts: the first part, $2z_{\alpha_m/2}\sigma_m$, is the same as the BEL of the $z$-based IE; and the second part, $\Phi(C_m)$, reflects the thresholding effect. From the form of $C_m$, it is evident that the BEL of the BIE Thres approaches the BEL of the $z$-based IE as $C$ goes to infinity. Next, we introduce the form and property of the BCP of the $m$th BIE Thres:
\begin{restatable}{prop}{propb}
	\label{prop2}
	The Bayes coverage probability (BCP) of the $m$th BIE Thres has the form of 
	\begin{align} \label{eq35}
	BCP[\mu_m,\alpha_m,C;\sigma_m,\tau]=2\int_{-\infty}^{C_m}\left\{\Phi\left(\tfrac{\sigma_m}{\tau}y+\sqrt{1+\tfrac{\sigma^2_m}{\tau^2}}z_{\alpha_m/2}\right)-\Phi\left(\tfrac{\sigma_m}{\tau}y\right)\right\}d\Phi(y)
	\end{align}
	where $C_m=C\tau/\sqrt{\sigma_m^2+\tau^2}$, and the BCP approaches the coverage probability of the corresponding $z$-based IE as $C$ goes to $\infty$.
\end{restatable}
The BCP of the BIE Thres has no closed form; nevertheless, it is twice differentiable with respect to $\alpha_m$, so a gradient-based optimization can be implemented in the next section. It is also not difficult to show that the BCP of the BIE Thres approaches $1-\alpha_m$, the BCP of the $z$-based IE.

\subsection{Global Performance Measures}
\leavevmode
Now, we investigate global performance measures of an BMIE Thres. All the measures are derived by assuming $M$ target parameters. First, we introduce a Bayes threshold ratio (BTR). It is the ratio of the number of the thresholded BIEs to the total number of the BIEs in an BMIE Thres, i.e., the proportion of the one-sided BIEs. The form of the BTR is presented in the following proposition:
\begin{restatable}{prop}{propc}
	\label{prop3}
	The Bayes threshold ratio (BTR) has the form of  
	\begin{align}\label{eq36}
	BTR[\bc{\mu},\bc{\alpha},C;\bc{\sigma},\tau]=\tfrac{2}{M}\sum_{i=1}^M\Phi(-C_m)
	\end{align}
	where $C_m=C\tau/\sqrt{\sigma_m^2+\tau^2}$. As $C$ increases from 0 to $\infty$, the BTR decreases from 1 to 0.  
\end{restatable}
In addition to the BTR, we have two more global measures which extend the BEL and BCP to the case with multiple parameters. For the global expected content, we already defined the BAEL in (\ref{eq32}). However, it is positively unbounded, hence it is not a meaningful measure of the quality of the MIE. Instead, we define a Bayes relative expected length (BREL), the ratio of the BAEL of an BMIE Thres to the BAEL of the corresponding $z$-based MIE. Then the BREL is always bounded between 0 and 1. For the global coverage probability, we utilize the BFWCR in (\ref{eq31}) based on the given normal-normal setting.   
\begin{restatable}{cor}{coro}
	\label{cor} Given $M$ parameters and a global level $1-q$, the Bayes relative expected length (BREL), Bayes family-wise coverage rate (BFWCR) and Bayes thresholding ratio (BTR) have the following forms:
	\begin{align*}
		BREL[\bc{\mu},\bc{\alpha},C;\bc{\sigma},\tau]=&\left.\frac{1}{M}\sum_{m=1}^M\left[2z_{\alpha_S}\sigma_m\Phi(C_m)\right]\middle/ \frac{1}{M}\sum_{m=1}^M\left[2z_{\alpha_S}\sigma_m\right]\right.;\\
		BFWCR[\bc{\mu},\bc{\alpha},C;\bc{\sigma},\tau]=&\prod_{m=1}^M\left[2\int_{-\infty}^{C_m}\left\{\Phi\left(\frac{\sigma_m}{\tau}y+\sqrt{1+\tfrac{\sigma^2_m}{\tau^2}}z_{\alpha_S}\right)-\Phi\left(\tfrac{\sigma_m}{\tau}y\right)\right\}d\Phi(y)\right];\\
		BTR[\bc{\mu},\bc{\alpha},C;\bc{\sigma},\tau]=&\frac{2}{M}\sum_{m=1}^M\Phi(-C_m)
	\end{align*}
	where $C_m=C\tau/\sqrt{\sigma_m^2+\tau^2}$, and $\alpha_S$ is the Sidak adjustment, $1-(1-q)^{1/M}$. As $C$ increases from 0 to $\infty$, the BREL, BFWCR, and BTR converge to 1, $1-q$, and 0, respectively.  
\end{restatable}
The global performance measures depend on the thresholding parameter $C$. Thus, it is worthwhile to visually ascertain their behaviors with respect to $C$. Figure \ref{GPQ} is based on an BMIE Thres for $M=1000$ normal location parameters under the global level $1-q=0.9$. For the setting, we assign an equi-spaced sequence from 0.01 to 10 for the standard errors, $\sigma_m$'s, and three different values 2, 3, and 5 for the prior standard deviation, $\tau$. Then we plot the graphs of global measures with respect to $C$ from 0 to 6.  

\begin{figure}[ht]
	\centering
	\includegraphics[scale=0.32]{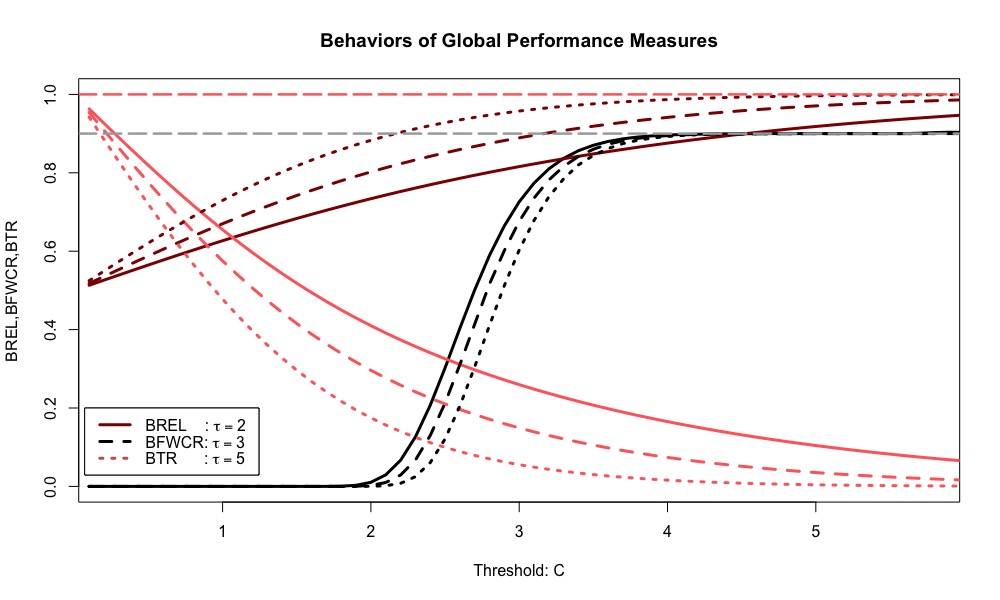}
	\caption[Three Global Measures of BMIE Thres]{Three Global Measures of BMIE Thres for different $\tau$'s} 
	\label{GPQ} 
\end{figure}
The solid, dashed, and dotted lines indicate the different $\tau$'s: 2, 3, and 5, respectively; and dark red, black, and light red color lines indicate the BREL, BFWCR, and BTR of the BMIE Thres.  
First, note that the opposite behaviors of the BTR and BREL. When $C$ equals zero, every BIE lacks an outer tail, so the BTR is one and the BREL is one half. As $C$ increases, the number of one-sided BIEs decreases, and this results in the decrease of the BTR and the increase of the BREL. The rates of decrease and increase depend on the value of $\tau$ since the smaller $\tau$ of the true distributions amounts to more information to remove one side of the BIEs in general. On the other hand, the BFWCR shows an interesting behavior. It remains almost zero up to a $C$ of around 2. This is because each BCP is much less than the given individual level; as a result, the BFWCR, i.e., the product of $M=1000$ BCPs, becomes almost zero. However, the BFWCR rapidly increases as $C$ passes the interval from 2 to 3, and it almost reaches the global level $1-q=0.9$ at around 3.5 for any $\tau$'s. When this occurs, we can clearly observe that the BREL is less than 1, implying the BAEL of BMIE Thres is smaller than that of $z$-based MIE. That is, the BMIE Thres provides a better precision while maintaining the same level of accuracy.

At this point, it is natural to ask about the optimal value for $C$. The optimal $C^*$ is the value which provides the smallest BREL of the BMIE Thres, while at the same time maintaining the BFWCR at least the global level. However, Figure \ref{GPQ} shows both the BREL and the BFWCR change as $C$ varies. Therefore, in order to determine the optimal thresholding parameter, $C^*$, we must first find a way to adjust the BFWCR to match the global level so that the BRELs can be properly compared throughout the values of $C$. This process can be achieved using the optimization method in the next section.

\section{Optimization}
In this subsection, we perform an optimization procedure which determines the best $C^*$ for an BMIE Thres. In the previous subsection, the BMIE Thres shows its limited ability in maintaining the BFWCR to be at least the global level when $C$ is small. The optimization procedure will provide a way to push up the BFWCR to the global level so that we can determine the optimal $C^*$ which provides the minimum BREL, or equivalently the minimum BAEL. Furthermore, the approach also allows us to assign optimal levels, $\alpha_m^*$'s, to individual BIEs. Here, we only describe the essence of the optimization method. Refer to Appendix A for details.

\subsection{Optimization Procedure}
\leavevmode
The optimization procedure consists of two global risk functions, $\boldsymbol{R_0^\beta}$ and $\boldsymbol{R_1}$, which represent the adjusted BREL and 1-BFWCR, respectively. We first set up the following optimization problem:  
\begin{align}\label{eq41}
\text{minimize}\;\; \boldsymbol{R^\beta_0(\theta,\delta)}\;\;\text{subject to}\;\;\boldsymbol{R_1(\theta,\delta)}\leq q
\end{align}
Here, the global risks can be expressed with the performance measures:
\begin{align}\label{eq42}
	\boldsymbol{R^\beta_0(\theta,\delta)}=&\frac{1}{M}\sum_{m=1}^M\frac{BEL[\mu_m,\alpha_m,C;\sigma_m,\tau]}{\beta+BEL[\mu_m,\alpha_m,C;\sigma_m,\tau]}; \nonumber\\
	\boldsymbol{R_1(\theta,\delta)}=&1-\prod_{i=1}^MBCP[\mu_m,\alpha_m,C;\sigma_m,\tau] 
\end{align}
where $C_m=C\tau/\sqrt{\sigma_m^2+\tau^2}$.
To make numerical implementation stable, we reparametrize $\alpha_m$'s to $\nu_m$'s, such that $\nu_m=z_{\alpha_m/2}=\Phi^{-1}\left(1-\alpha_m/2\right)$. Then the restated optimization problem is as follows:
\begin{align*}
	\text{minimize}&\;\; \frac{1}{M}\sum_{m=1}^M\frac{BEL[\mu_m,\nu_m,C;\sigma_m,\tau]}{\beta+BEL[\mu_m,\nu_m,C;\sigma_m,\tau]}\\
	\text{subject to}&\;\;\sum_{m=1}^M\log\left( BCP[\mu_m,\nu_m,C;\sigma_m,\tau]\right)\geq \log(1-q)
\end{align*}
\[
\text{where }
\begin{cases}
BEL[\mu_m,\nu_m,C;\sigma_m,\tau]=&2\nu_m\sigma_m\Phi(C_m);\\
BCP[\mu_m,\nu_m,C;\sigma_m,\tau]=&2\int_{-\infty}^{C_m}\left\{\Phi\left(\tfrac{\sigma_m}{\tau}z_m+\tfrac{\sqrt{\tau^2+\sigma^2_m}}{\tau}\nu_m\right)-\Phi\left(\tfrac{\sigma_m}{\tau}z_m\right)\right\}d\Phi(z_m).
\end{cases}
\]
This optimization problem is solvable because $BEL[\mu_m,\nu_m,C;\sigma_m,\tau]$ and $BCP[\mu_m,\nu_m,C;\sigma_m,\tau]$ are at least twice differentiable with respect to $\nu_m$ for any fixed $C$. Thus, we can set up a Lagrange equation, and solve it numerically by using the Newton-Raphson method. 

\subsection{Optimization Result}
\leavevmode
The result of the optimization is presented in this subsection. We exploit the setting for Figure \ref{GPQ} in which we ascertained the behaviors of the global measures. After the optimization, the BFWCRs are always equal to the global level, $1-q=0.9$, for any given $\tau$ and $C$. Therefore, we only present the resulting BRELs for different $\tau$'s, 2, 3, and 5, with respect to $C$ in Figure \ref{OPT}.
\begin{figure}[ht]
	\centering
	\includegraphics[scale=0.42]{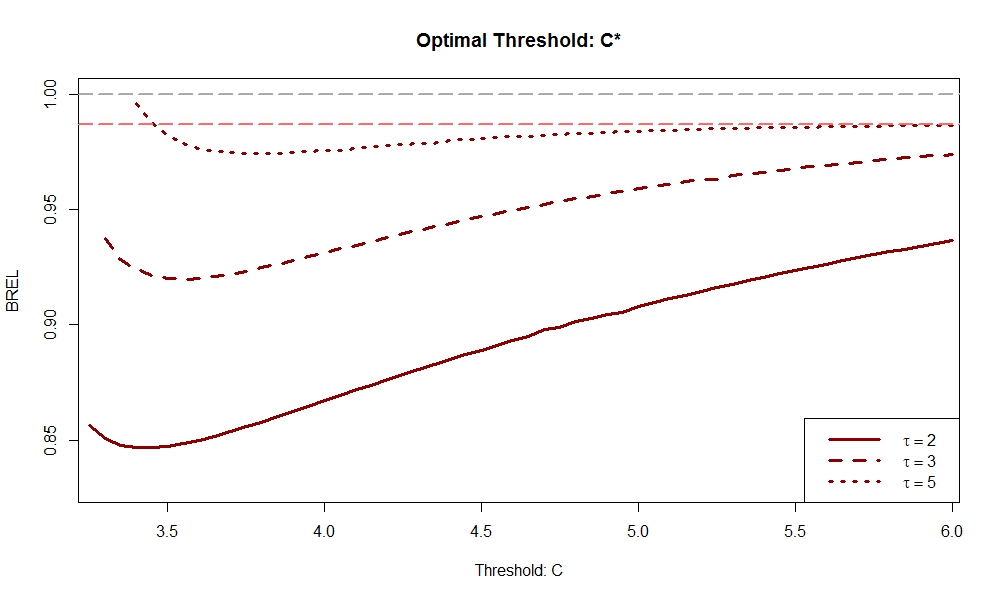}
	\caption[BREL after the Optimization]{$C^*$ obtained from Optimization} 
	\label{OPT} 
\end{figure}
In the plot, the trajectories show similar bath-tub shapes on the left hand side. The lowest points represent the smallest BRELs which determine the optimal $C^*$'s. For the different $\tau$'s, 2, 3, and 5, the minimum BRELs are 84.7\%, 92.0\%, and 97.4\%, with the corresponding optimal $C^*$'s being equal to 3.4, 3.5, and 3.8, respectively. For each optimal $C^*$, the corresponding BREL increases as $C$ increases or decreases. When $C$ increases, the increase of the BREL is evident due to the decrease of the BTR, i.e., the ratio of one-sided BIEs. When $C$ decreases, the increase of the BREL is due to the optimization procedure which pushes up the BFWCR to the global level $1-q=0.9$. This is also apparent from a single dimensional observation wherein we get a wider expected interval length by increasing its confidence level. 

It is worth mentioning that the BRELs converge to the light red dashed line instead of 1, the gray dashed line. This is because the reduction between the gray and light red dashed lines is solely achieved by the optimal allocation of the individual levels, $\alpha_m^*$'s, as shown in Appendix A.
That is, the amount of reduction can also be obtained from the optimization procedure with the classical $z$-based MIE. Therefore, the rest of the reductions between the minimum BRELs and the light red dashed line measures the actual thresholding effect of the BMIE Thres. Figure \ref{OLA} shows the optimal individual level allocation result with respect to $C$ for $\tau=3$. 
\begin{figure}[ht]
	\centering
	\includegraphics[scale=0.42]{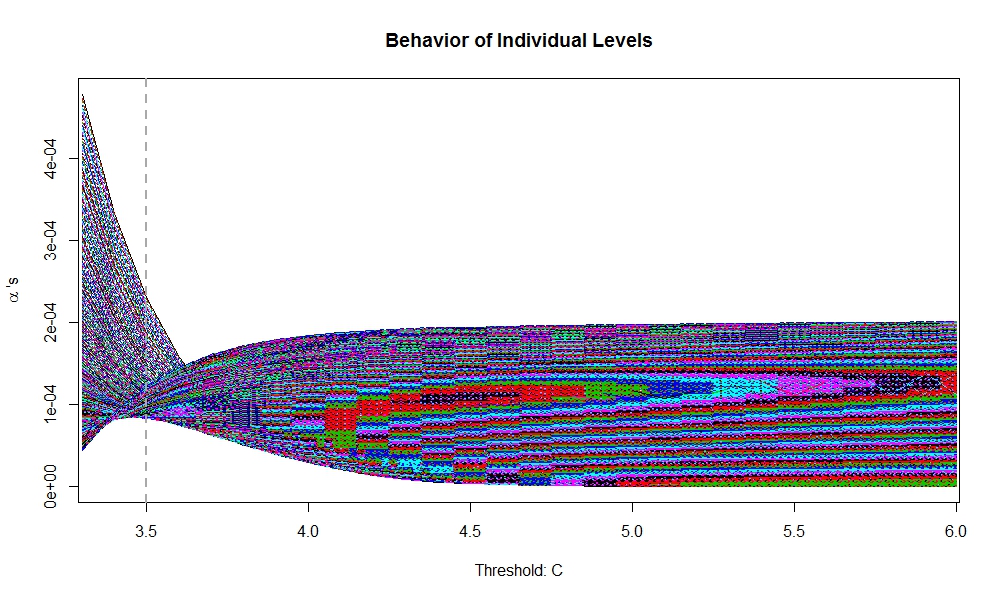}
	\caption[Optimal Level Allocation]{Behavior of Individual Levels, $\alpha_m$'s, for $\tau=3$} 
	\label{OLA} 
\end{figure}
Note that the optimal $C^*=3.5$ is the value where the behaviors of individual levels $\alpha_m$'s become reversed. In conclusion, the result suggests that meaningful amounts of the BREL reduction can be achieved when $\tau$ is relatively small, and the reduction is contributed by two parts: the reduction from the thresholding which can be optimized by $C^*$ and the reduction from the optimal individual level allocation, $\alpha^*_m$'s.

\section{Data Applications}
\subsection{In-season Batting Average Data: Known $\sigma_m$'s}
\leavevmode
In-season prediction of batting averages has been studied by several researchers, e.g., \citet{efron:1975} and \citet{brown:2008}, for a couple of reasons. First, while the beginning portion of baseball season data is used for the prediction, the batting average from the remainder of the season can be considered as a success probability, $p$, which reflects player's true batting capability for the season. Therefore, one can set up a binomial model with the $p$ along with another parameter, $n$, the total number of at-bats. Second, when the arcsine transformation is utilized to convert the binomial distribution into an approximate normal distribution, the standard deviation is derived solely from the number of at-bats that can be obtained from the data. Therefore, the individual problem can be set up with a normal distribution with a known standard deviation, so that we can utilize the form of BMIE Thres procedure in (\ref{eq33}), by simply replacing the hyper-parameters and individual levels with their estimated and optimized values, respectively.  

The previous studies on the in-season prediction of batting averages utilized the parametric empirical Bayes approach to predict the true batting ability by using the beginning portion of the data. For example, \citet{efron:1975,efron:1977} used the batting average of the first 45 at-bats in the 1970 MLB season and \citet{brown:2008} used the batting average of the first three months' data in the 2005 MLB season for the predictions. These previous studies concentrated on the point estimations, with the performance evaluated by the total squared error: $\sum(x_m-\theta_m)^2$, where $\theta_m$ is the batting average of the $m$th player and $x_m$ is the corresponding estimate for $\theta_m$. In this study, we focus on constructing an BMIE Thres and evaluating its performance by ascertaining the BREL and BFWCR. We utilize the 2005 MLB season data which is also used in \citet{brown:2008}, as it involves a larger number of players than the data in \citet{efron:1975}. 

\subsubsection{Problem Setup and Assumptions}
\leavevmode
Let $H_m$ and $N_m$ be the number of hits and at-bats for the $m$th player over the whole 2005 season. As we choose the first $j$ month(s) for the prediction period, $H^j_m$ and $N^j_m$ are defined as the number of hits and at-bats of the first $j$ month(s) period for $j=1,\; 2,\; \text{or},\; 3$. Once $j$ has been determined, the notation for $j=4$ is reserved for the number of hits and at-bats for the rest of the season. In this application, we exclude pitchers as well as batters who have fewer than 11 at-bats in either inside or outside of the prediction period. Therefore, each prediction period contains a different number of players, $M_j$. 
Given any index $j$ for the prediction period, we first set up the following binomial model:
\begin{align} \label{eq51}
H_m\sim Binom(N_m,p_m),\;\;\text{for}\;\;m=1,2,\ldots,M.
\end{align}
Then the arcsine transformation is utilized to obtain the following approximate normal model, and a normal prior distribution is added:
\begin{align} \label{eq52}
X_m=\arcsin\sqrt{\frac{H_m+1/4}{N_m+1/2}}\stackrel{\text{approx.}}{\sim}\mathcal{N}(\mu_m,\sigma^2_m)
\;\;\&\;\;\mu_m\sim\mathcal{N}(\eta,\tau^2)
\end{align}
where $\mu_m=\arcsin\sqrt{p_m}$ and $\sigma_m^2=\frac{1}{4N_m}$.
In addition, we assume that $\mu^j_m=\mu_m$ for $j=1,\; 2,\; 3,\; \text{and}\; 4$, meaning the true batting average of the $m$th player does not change throughout the season. Therefore, the observed batting average value from the rest of the season, $x^4_m$, is utilized for the true batting average $\mu_m=\mu_m^4$.

\subsubsection{Hyper-parameter Estimation}
\leavevmode
In order to implement the BMIE Thres, the hyper-parameters, $\eta$ and $\tau$, need to be estimated. We follow the ML-II approach in \citet{good:1983}. This estimation procedure considers the expression of the marginal density of $X^j_m$, $\mathcal{N}(\eta,(\sigma^j_m)^2+\tau^2)$, as a likelihood function of the hyper-parameters and seeks the values which maximize the likelihood. The resulting estimators can be obtained by solving the following two equations:
\begin{align} \label{eq53}
\hat{\eta}=\frac{\sum_{m=1}^{M_j} X^j_m/({\sigma^j_m}^2+\hat{\tau}^2)}{\sum_{m=1}^{M_j}1/({\sigma^j_m}^2+\hat{\tau}^2)}\;\;\&\;\;\sum_{m=1}^{M_j}\frac{(X^j_m-\hat{\eta})^2}{({\sigma^j_m}^2+\hat{\tau}^2)^2}=\sum_{m=1}^{M_j}\frac{1}{{\sigma^j_m}^2+\hat{\tau}^2}.
\end{align}
As the estimation is quite sensitive to the initial values of $\hat{\eta}$ and $\hat{\tau}$, we started from the following initial values:
\begin{align} \label{eq54}
\hat{\eta}_{\text{ini}}=\frac{1}{M_j}\sum_{m=1}^{M_j}X^j_m\;\;\&\;\;
\sum_{m=1}^{M_j}\frac{(X^j_m-\hat{\eta}_\text{ini})^2}{({\sigma^j_m}^2+\hat{\tau}_\text{ini}^2)^2}=\sum_{m=1}^{M_j}\frac{1}{{\sigma^j_m}^2+\hat{\tau}_\text{ini}^2}.
\end{align}

\subsubsection{Optimal $C^*$ and $\alpha^*_m$'s}
\leavevmode
The optimal $C^*$ and $\alpha^*_m$'s can be obtained through the optimization procedure above. Note that the procedure is performed based on the $\sigma_m$'s and the estimated hyper-parameters $\hat{\eta}$ and $\hat{\tau}$. In addition, the tuning parameter $\beta$ in the optimization is fixed to be 1000 as in Appendix A.
With the optimized values, the form of the BIE Thres for the $m$th player with the $j$th prediction period becomes:
\begin{align}\label{eq55}
\Gamma_m(X^j_m;\mu^j_m,\alpha_m^*)=\left(X^j_m-z_{\alpha^*_m/2} \sigma^j_mI\left\{X^j_m>\hat{\eta}^j-C^*\hat{\tau}^j\right\},X^j_m+z_{\alpha^*_m/2}\sigma^j_mI\left\{X^j_m<\hat{\eta}^j+C^*\hat{\tau}^j\right\}  \right).
\end{align}

\subsubsection{Performance of BMIE Thres on Batting Average Prediction}
\leavevmode
We implement the procedure for different prediction periods: April, April-May, and  April-June. In each case, the values of target parameter $\mu_m$'s are obtained from the data for the rest of the season, respectively: May-October, June-October, and July-October. The global level is set to be $1-q=0.9$. The results of this illustrative application are summarized in Table \ref{AR}.
\begin{table}[ht]
	\centering
	\caption{Result for Batting Averages Prediction}
	\small{
		\begin{tabular}{|c|c:c:c|}
			\hline
			Prediction Period	& April (j=1) & April$\sim$May (j=2) & April$\sim$June (j=3)    \\
			\hline
			$\hat{\eta}$  & 0.5425 & 0.5438 & 0.5468      \\
			$\hat{\tau}$  & 0.008 & 0.0123 & 0.0151     \\
			$C^*$       & 2.86 & 3.04  &   3.195  \\
			\hline
			BFWCR & 92.24\% (100\%) & 95.37\% (98.78\%) &    93.79\% (97.70\%)  \\
			BREL & 62.58\% & 76.46\%   & 82.95\%   \\
			BTR & 73.39\% & 42.68\%  & 28.05\%  \\
			\hline
			$M_j$    & 387 & 410  & 435      \\
			\hline
		\end{tabular}}
		\label{AR}
	\end{table}
	First, the BFWCRs are consistently higher than the global level $1-q$ in all of the prediction periods. Note that the corresponding BFWCRs of the classical $z$-based MIEs are in parentheses. The values are quite larger than the global level, implying the $z$-based MIEs are unnecessarily conservative. Now, note that the estimated $\tau$ increases as the prediction periods become wider. This prior information determines the optimal value $C^*$. As a result, we have the smallest $C^*$ when $j=1$, so the BTR becomes high and the corresponding BREL, 62.58\%, shows a considerable reduction. When $j=2$, we have a larger $C^*$, and this results in a higher BREL of 76.45\%. Lastly, we have the highest $C^*$ when $j=3$, so that the highest BREL, 82.95\%, is obtained. In conclusion, the BMIE Thres consistently shows reasonably higher-than-nominal-level BFWCRs; at the same time, it achieves meaningful reductions on the BRELs. However, the amount of reduction depends on the estimated value of $\tau$ which governs the prior information on the target parameter.

\subsection{Leukemia Gene Expression Data: Unknown $\sigma_m$'s}
\leavevmode
The leukemia data appeared in \citet{efron:2016} as a type of gene expression data from high-density oligonucleotide microarrays. It consists of $n=72$ patients with $n^1=45$ of ALL (acute lymphoblastic leukemia) group and $n^2=27$ of AML (acute myeloid leukemia) group, which has a worse prognosis. Efron provides small and large data sets which contain $M_1=3571$ and $M_2=7128$ genes, respectively. To eliminate response disparities among the $M$ microarrays as well as some outliers, the raw expression levels for the $m$th gene on the $k$th patient, $X_{mk}$, were transformed to a normal score value, 
\begin{align}\label{eq56}
	x_{mk}=\Phi^{-1}\left(\frac{rank(X_{mk})-0.5}{M}\right).
\end{align} 
Efron's investigation of the data was about a multiple testing procedure based on the local-FDR (see \citet{efron:2012}). However, our goal here is to construct the BMIE Thres for the mean difference between the ALL and AML groups and to compare this result with the classical $t$-based MIE. 
Note that we equally compare the two group means in this application. Therefore, zero value is utilized to evaluate the empirical coverage probability.

\subsubsection{Problem Setup}
\leavevmode
The data has the form of a $M\times n$ matrix where $M$ is either 3571 or 7128 and $n=n^1+n^2$ is $72=45+27$. Next, $x_{mk}$ is the expression level for the $m$th gene of the $k$th patient, where $m=1,2,\ldots,M$, $k=1,2,\ldots,n^1$, for the AML group and $k=46,47,\ldots,n$, for the ALL group. This is the case of unknown standard deviations with one of the sample sizes less than 30, so that we could apply a plug-in procedure with the corresponding sample standard deviations. Thus, the basic individual procedure follows the two-sample $t$-based interval estimation under the equal variances assumption. (\citet{lehmann:2006})

\subsubsection{Hyper-Parameter Estimation and Optimal Threshold}
\leavevmode
For the hyper-parameter estimation, we can follow the same ML-II procedure as in the case of known $\sigma_m$'s in the previous application. However, it becomes hard to perform the optimization since the Newton-Raphson method does not work with quite large $M$s in this application. Therefore, we instead utilize a graphical search for the optimal $C^*$ as in Figure \ref{GS}.
	\begin{figure}[ht]
		\centering
		\includegraphics[scale=0.25]{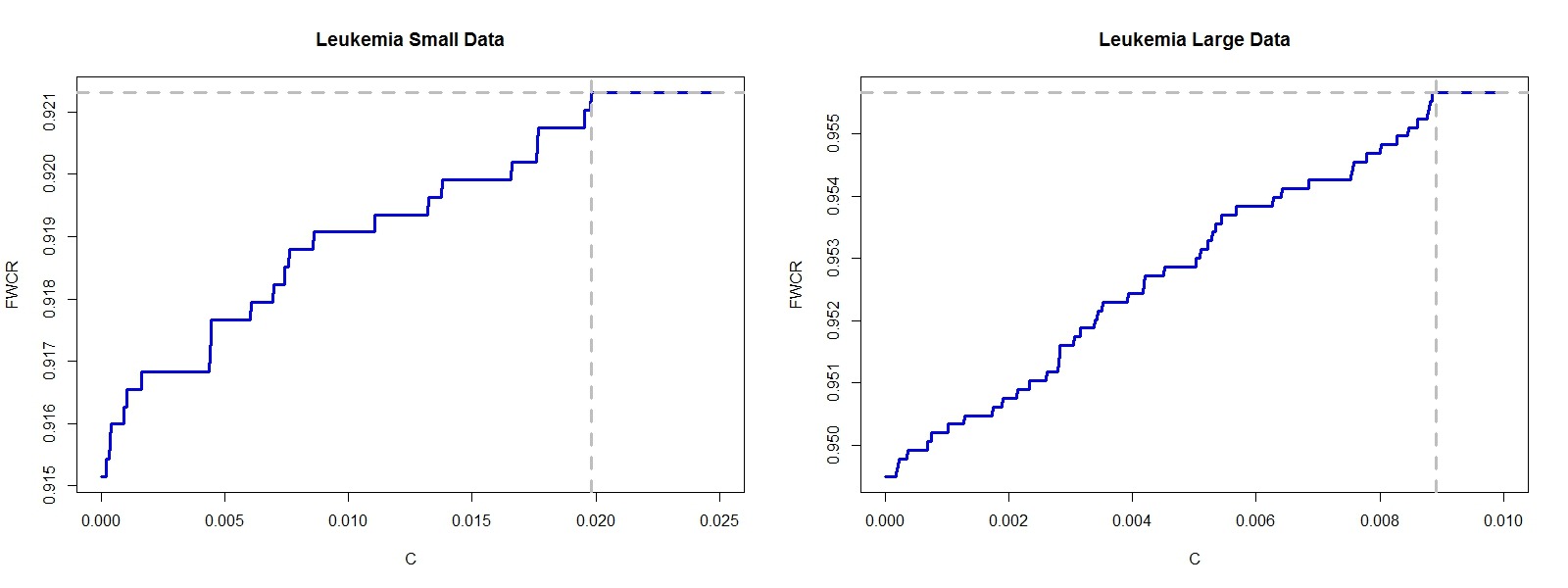}
		\caption[Graphical Search for $C^*$ based on BFWCRs]{Graphical Search for $C^*$ based on BFWCRs} 
		\label{GS} 
		\vspace{-10pt}
	\end{figure}
Note that given a global level, $1-q=0.9$, the BFWCRs of the classical $t$-based MIE for small and large leukemia data sets are 0.921 and 0.956, respectively. Since the BFWCRs should converge to the above values as $C$ increases, we choose the optimal $C^*$ at the points of intersections: $C^*=0.0198$ for the small data and $C^*=0.0098$ for the large data. These values are the smallest values which have the same BFWCRs as those of the classical MIEs. By the same token, the Sidak adjustment, $\alpha_S=1-(1-q)^{1/M}$, is applied for the individual levels since it is available without the optimization.
Then the form of the $m$th BIE Thres is as follows:
	\begin{align}\label{eq57}
		\Gamma_m(\bc{X}_m;,\mu_m\alpha_S)=&\left(\bar{X}^{1}_m-\bar{X}^{2}_m-t_{\alpha_S/2,n^1_m-n^2_m-2}S^p_m\sqrt{1/n^1_m+1/n^2_m}I\left\{\bar{X}_m>\hat{\eta}-C^*\hat{\tau}\right\},\right. \nonumber\\
		&\left.
		\bar{X}^{1}_m-\bar{X}^{2}_m+t_{\alpha_S/2,n^1_m-n^2_m-2}S^p_m\sqrt{1/n^1_m+1/n^2_m}I\left\{\bar{X}_m<\hat{\eta}+C^*\hat{\tau}\right\}  \right)
	\end{align}	
where $S^p_m=\sqrt{\frac{\sum_{k=1}^{n_m^1}(x_{mk}-\bar{x}^1_m)^2+\sum_{k=n_m^1+1}^{n_m}(x_{mk}-\bar{x}^2_m)^2}{n_m^1+n_m^2-2}}$ and $t_{\alpha_S/2,n^1_m+n^2_m-2}$ is the $1-\alpha_S/2$ quantile of a $t$-distribution with $n^1_m+n^2_m-2$ degrees of freedom.
Due to the use of the Sidak adjustment, we lose the opportunity to assign the optimal levels into the individual IEs. However, this will not cause problems in this particular illustration, as we will see in the next subsection.

\subsubsection{Performance of BMIE Thres on Leukemia Data}
\leavevmode
In this subsection, we construct the BMIE Thres under the global level $1-q=0.9$. The result is summarized in Table \ref{LD}.
\begin{table}[ht]
	\centering
	\caption{Result for Leukemia Data}
	\small{
		\begin{tabular}{|c|c:c|}
			\hline
			Leukemia 	& Small Data & Large Data \\
			\hline
			$\hat{\eta}$  & 0.0108 & 0.0014 \\
			$\hat{\tau}$  & 0.5336 & 0.1598 \\
			$C^*$            & 0.0198 & 0.0098 \\
			\hline
			BFWCR & 92.13\% & 95.57\%    \\
			BREL & 50.65\% & 50.46\%     \\
			BTR & 98.74\% & 98.81\%    \\
			\hline
			$M$    & 3571 & 7128       \\
			\hline
		\end{tabular}}
		\label{LD}
	\end{table}
By design, the BFWCRs of the BMIE Thres are the same as those of the classical $t$-based MIEs. However, the BRELs of the BMIE Thres are less than 51\% in both cases, implying strong performance in reducing the global expected content. These seemingly \textit{too good} results can be justified by the motivational sketch of the thresholding approach. 
	\begin{figure}[ht]
		\centering
		\includegraphics[scale=0.35]{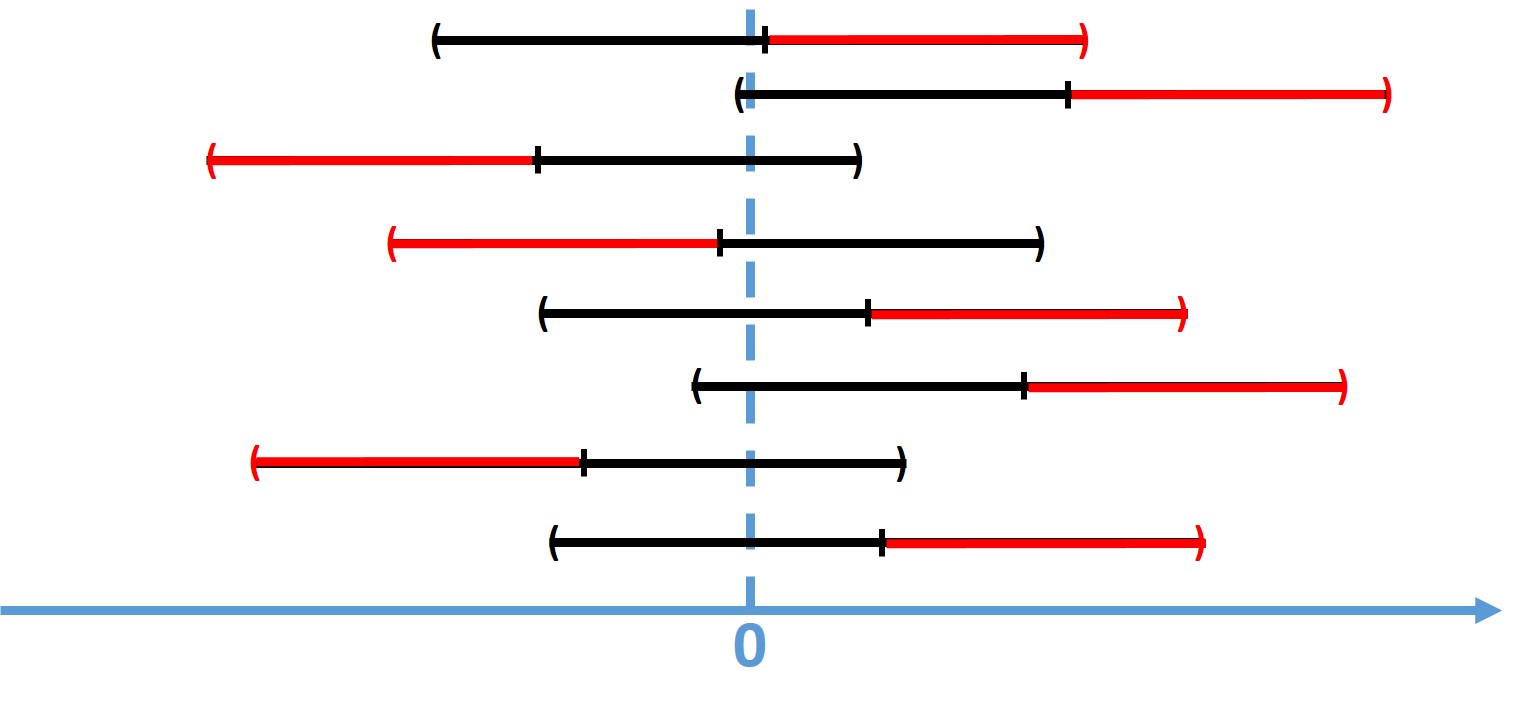}
		\caption[Bayes FWCR for Leukemia Data]{Comparing Two Population Means} 
		\label{CTP} 
		\vspace{-10pt}
	\end{figure}
As mentioned, the thresholding procedure removes the outer tails of BIEs. However, with regard to comparing two group means, the inner tails cover zeros in most cases as in Figure \ref{CTP}. In some sense, the thresholding scheme of the BMIE Thres is designed to operate well in this particular application of comparing two population means by discarding a number of redundant outer tails. Because of the significant reductions already achieved, we do not have to exploit the additional reduction from the formal optimization procedure. This can also be an advantage when we consider the potential applicability of the BMIE Thres to high-throughput data set.

\section{Simulation}
We ascertained the performance of BMIE Thres by investigating its theoretical global measures in subsection 3.4. However, the illustration was based on an ideal situation in which the suggested model perfectly represented the true underlying structure as follows:
\[
\bar{X}_m|\mu_m\sim \mathcal{N}\left(\mu_m,\sigma^2_m\right)\;\;\&\;\;\mu_m\sim \mathcal{N}^\Pi\left(\eta,\tau^2\right)=\mathcal{F}^*(\eta^*,\tau^*)\;\text{for}\;m=1,2,\ldots,M,
\]
where $\mathcal{N}^\Pi(\eta,\tau^2)$ was the prior distribution in our minds and $\mathcal{F}^*(\eta^*,\tau^*)$ was the true generating distribution, which the target parameter $\mu_m$'s followed. By assuming the two distributions were identical, we could effectively present the general behavior of the BMIE Thres as in Figure \ref{GPQ}. However, this assumption is unrealistic in the sense that it does not encompass the case of misspecified models. This is because no one can guarantee that the prior distribution is identical to the true generating distribution, i.e., it is always possible $\mathcal{N}^\Pi\ne \mathcal{F}^*$. Moreover, even though the distributions are identical, the parameters can still be distinct from one another, i.e., it is also possible $(\eta,\tau)\ne (\eta^*,\tau^*)$. Therefore, this simulation study is designed to emulate situations of prior misspecification. That is, we assume the true generating distribution, $\mathcal{F}^*$, to be normal, uniform, logistic, or exponential, with a fixed true mean, $\eta^*$, and standard deviation, $\tau^*$, in order to generate $\mu_m$'s. However, the prior distribution assigned is always the normal distribution with the hyper-parameters $\eta$ and $\tau$, and we set these to be different from the true parameters.   

In addition, since the proposed \textit{frequentist} procedure, BMIE Thres, depends on the prior distribution, it would be reasonable to ascertain its performance in relation to the performance of the Bayesian credible MIE or the classical $z$-based MIE. In involving these comparisons in the simulation, we will consider a total of five MIEs: $z$-based classical MIE, $\bc{\Gamma_0(X)}$; BMIE Thres with hyper-parameter, $\bc{\Gamma_1(X)}$; BMIE Thres with estimated hyper-parameters, $\bc{\Gamma_2(X)}$; Bayesian credible MIE with hyper-parameters, $\bc{\Gamma_3(X)}$; and Bayesian credible MIE with estimated hyper-parameters, $\bc{\Gamma_4(X)}$. 
To amplify, even though there is a true stochastic model generating the $\mu_m$'s, our main inferential interest are the values of the generated $\mu_m$'s, which are {\em the} parameters of interest, instead of the generating stochastic model. As such, the evaluation of the procedures are with respect to the true values of these parameters.
The $m$th individual IEs of the five MIEs are as follows: 
\begin{align*}
	\Gamma^0_m(\bc{X})&=\left(\bar{X}_m-z_{\alpha_s/2}\sigma_m,\bar{X}_m-z_{\alpha_s/2}\sigma_m\right);\\
	\Gamma^1_m(\bc{X})&=\left(\bar{X}_m-z_{\alpha_m/2}\sigma_mI\left\{\bar{X}_m>\eta-C^*\tau\right\},\bar{X}_m+z_{\alpha_m/2}\sigma_mI\left\{\bar{X}_m<\eta+C^*\tau\right\}  \right);\\
	\Gamma^2_m(\bc{X})&=\left(\bar{X}_m-z_{\alpha_m/2}\sigma_mI\left\{\bar{X}_m>\hat{\eta}-C^*\hat{\tau}\right\},\bar{X}_m+z_{\alpha_m/2}\sigma_mI\left\{\bar{X}_m<\hat{\eta}+C^*\hat{\tau}\right\}  \right);\\
	\Gamma^3_m(\bc{X})&=\left(\beta\bar{X}_m+(1-\beta)\eta-z_{\alpha_s/2}\sigma_m\sqrt{\beta},\beta\bar{X}_m+(1-\beta)\eta+z_{\alpha_s/2}\sigma_m\sqrt{\beta}\right);\\
	\Gamma^4_m(\bc{X})&=\left(\hat{\beta}\bar{X}_m+(1-\hat{\beta})\hat{\eta}-z_{\alpha_s/2}\sigma_m\sqrt{\hat{\beta}},\hat{\beta}\bar{X}_m+(1-\hat{\beta})\hat{\eta}+z_{\alpha_s/2}\sigma_m\sqrt{\hat{\beta}}\right)
\end{align*}
where $\beta=\frac{\tau^2}{\sigma^2+\tau^2}$, $\hat{\beta}=\frac{\hat{\tau}^2}{\sigma^2+\hat{\tau}^2}$, and $\alpha_s=1-(1-q)^{1/M}$, the Sidak adjustment. Note that the last MIE is also called the empirical Bayes MIE under the normal-normal model.  

The goal of the simulation is to perform simultaneous interval estimation on $M=1000$ normal location parameters, $\mu_m$'s, given the global level $1-q=0.9$. These $\mu_m$'s are generated from different true distributions -- normal, uniform, logistic, and exponential -- with the fixed true mean $\eta^*$ and standard deviation $\tau^*$. To construct $\Gamma_1$ and $\Gamma_3$, we directly utilize the prior mean $\eta$ and standard deviation $\tau$, and these values will deviate from the true mean and standard deviation. To construct $\Gamma_2$ and $\Gamma_4$, we utilize estimated prior mean $\hat{\eta}$ and standard deviation $\hat{\tau}$ obtained from the data. The estimated prior standard deviation is also used to determine the optimal threshold $C^*$ and $\alpha_m^*$'s. The simulation scheme is provided by the pseudo code in Algorithm \ref{ALGO}.

\begin{algorithm}
	\caption{Simulation Scheme}
	\begin{algorithmic}
		\STATE \#\# $M=1000$ and $\sigma_m$'s are from $unif(0.01,10)$
		\FOR{$i=1$ in 1:4}
		\FOR{$j=1$ in 1:3}
		\FOR{$k=1$ to $Nrep=1000$}
		\STATE Data$\leftarrow$ Generator($\eta^*[i]$,$\tau^*[j]$)
		\STATE $C^*$$\leftarrow$Optimizer(Data,$\hat{\tau}[j]$)
		\STATE MIE$\leftarrow$Constructor(Data,$C^*$,$\eta[i]$,$\tau[j]$,$\hat\eta[i]$,$\hat\tau[j]$)
		\STATE Out[k]$\leftarrow$Evaluator(MIE,$\eta^*[i]$,$\tau^*[j]$)
		\ENDFOR	
		\STATE Result$[[i]][[j]]\leftarrow$Summarizer(Out)
		\ENDFOR
		\ENDFOR
		\STATE Tabulator(Result)
	\end{algorithmic}
	\label{ALGO}
\end{algorithm}

We first consider the case in which the true generating distribution is normal,
$\mathcal{N}^\Pi=\mathcal{F}^*$, but the hyper-parameters deviate from the true mean and standard deviation, $(\eta,\tau)\ne (\eta^*,\tau^*)$. That is, the true mean and standard deviation are fixed to be 0 and 2, but $\eta$ and $\tau$ take the values (0, 2, 4, 6) and (1, 2, 3), respectively. When we compare the performances of five MIEs, the empirical global coverage probability is obtained by the ratio of the IEs which cover the true target parameters in MIEs, and the empirical global expected content is measured by the relative average expected length of MIEs, i.e., the average expected lengths of the MIEs in relation to the average expected length of the corresponding $z$-based MIE.

\begin{table}[ht]
	\centering
	\caption{Prior Misspecification when the True Distribution is Normal}
	\small{
		\begin{tabular}{|c|c|c:c:c:c|c:c:c:c|}
			\hline
			\multicolumn{2}{|c|}{Normal}    & \multicolumn{4}{c|}{Global Coverage Probability} & \multicolumn{4}{c|}{Global Expected Content}  \\
			\hline
			$1-q=0.9$ &\scriptsize{MIEs} & \scriptsize{$\eta^*=0$}           & \scriptsize{$\eta=2$}          & \scriptsize{$\eta=4$}          & \scriptsize{$\eta=6$}          & \scriptsize{$\eta^*=0$}           & \scriptsize{$\eta=2$}          & \scriptsize{$\eta=4$}          & \scriptsize{$\eta=6$}                \\
			\hline
			\multirow{4}{*}{$\tau=1$}
			&$\Gamma_0$&0.896&0.909&0.911&0.873&1.000&1.000&1.000&1.000\\   
			&$\Gamma_1$&0.904&0.902&0.810&0.082&0.898&0.910&0.895&0.861\\
			&$\Gamma_2$&0.904&0.903&0.912&0.882&0.898&0.893&0.902&0.918\\
			&$\Gamma_3$&0.000&0.000&0.000&0.000&0.181&0.181&0.181&0.181\\
			&$\Gamma_4$&0.878&0.879&0.870&0.878&0.328&0.327&0.327&0.327\\
			\hline
			\multirow{4}{*}{$\tau^*=2$} 
			&$\Gamma_0$&0.909&0.903&0.922&0.883&1.000&1.000&1.000&1.000\\
			&$\Gamma_1$&0.909&0.905&0.922&0.888&0.915&0.927&0.951&0.946\\
			&$\Gamma_2$&0.909&0.905&0.922&0.888&0.914&0.902&0.917&0.906\\
			&$\Gamma_3$&0.894&0.271&0.000&0.000&0.328&0.328&0.328&0.328\\
			&$\Gamma_4$&0.861&0.860&0.867&0.860&0.326&0.326&0.326&0.327\\
			\hline
			\multirow{4}{*}{$\tau=3$} 
			&$\Gamma_0$&0.887&0.897&0.907&0.885&1.000&1.000&1.000&1.000\\		  
			&$\Gamma_1$&0.887&0.895&0.907&0.885&0.909&0.925&0.941&0.940\\
			&$\Gamma_2$&0.887&0.895&0.907&0.885&0.914&0.879&0.910&0.902\\
			&$\Gamma_3$&0.981&0.983&0.832&0.160&0.446&0.446&0.446&0.446\\
			&$\Gamma_4$&0.866&0.883&0.888&0.865&0.326&0.327&0.327&0.327\\
			\hline      
		\end{tabular}}
		\label{sim_norm}
	\end{table}
Table \ref{sim_norm} shows the simulation results under the correctly specified prior up to its distributional level. When the prior mean ($\eta$) and prior standard deviation ($\tau$) are equal to the true values, $\eta^*=0$ and $\tau^*=2$, all the MIEs perform well, showing the satisfactory global coverage probability and the reductions in the global expected contents. In particular, the global expected contents of the Bayes credible MIEs, $\Gamma_3$ and $\Gamma_4$, are significantly smaller than those of the $z$-based MIE and the BMIE Thres. However, as $\eta$ deviates from the true mean, the coverage of $\Gamma_3$ rapidly decreases; still, $\Gamma_4$ performs well with the estimated hyper-parameters except for a slight degradation in the global coverage rate. When the prior standard deviation is specified to be less ($\tau=1$) than the true standard deviation, the global coverage probabilities are affected by the concentrated prior distribution, resulting in substantial degradation. In this situation, $\Gamma_1$ is also affected by the concentrated prior, showing a gradual degradation in the global coverage probability as $\eta$ deviates from the true mean. The global expected contents of MIEs are generally narrower than in the previous case, reflecting the concentrated prior information. Lastly, when the prior standard deviation is specified to be larger ($\tau=3$) than the true standard deviation, the global level requirements are generally well satisfied due to the effect of the diffused prior, except for $\Gamma_3$ which shows rapid decrease as $\eta$ deviates. 
To compensate, the global expected contents become larger in general. In terms of the global coverage probability, the MIEs with the estimated parameters, $\Gamma_2$ and $\Gamma_4$, show the robust result as expected. Among those, $\Gamma_4$ performs very well compared to $\Gamma_2$ in terms of the global expected content. As a result, when the prior distribution is close to the true underlying structure up to its distributional level, $\Gamma_4$ would be the best choice if one can endure a slight degradation of the global coverage probability; however, if the global level requirement needs to be strictly satisfied, then $\Gamma_2$ would be the choice as it shows a very robust result for any hyper-parameter combination.
	
	\begin{table}[ht]
		\centering
		\caption{Prior Misspecification when the True Distribution is Uniform}
		\small{
			\begin{tabular}{|c|c|c:c:c:c|c:c:c:c|}
				\hline
				\multicolumn{2}{|c|}{Uniform}    & \multicolumn{4}{c|}{Global Coverage Rate} & \multicolumn{4}{c|}{Global Content}  \\
				\hline
				$1-q=0.9$ &\scriptsize{MIEs} & \scriptsize{$\eta^*=0$}           & \scriptsize{$\eta=2$}          & \scriptsize{$\eta=4$}          & \scriptsize{$\eta=6$}          & \scriptsize{$\eta^*=0$}           & \scriptsize{$\eta=2$}          & \scriptsize{$\eta=4$}          & \scriptsize{$\eta=6$}                \\
				\hline
				\multirow{4}{*}{$\tau=1$}   
				&$\Gamma_0$&0.902&0.893&0.884&0.889&1.000&1.000&1.000&1.000\\
				&$\Gamma_1$&0.900&0.892&0.884&0.889&0.712&0.789&0.829&0.847\\
				&$\Gamma_2$&0.900&0.892&0.884&0.889&0.736&0.746&0.746&0.746\\
				&$\Gamma_3$&0.522&0.000&0.000&0.000&0.181&0.181&0.181&0.181\\
				&$\Gamma_4$&0.985&0.980&0.979&0.981&0.327&0.327&0.327&0.327\\		
				\hline
				\multirow{4}{*}{$\tau^*=2$} 
				&$\Gamma_1$&0.928&0.895&0.902&0.911&1.000&1.000&1.000&1.000\\		
				&$\Gamma_1$&0.920&0.895&0.902&0.910&0.718&0.793&0.832&0.850\\
				&$\Gamma_2$&0.920&0.895&0.902&0.910&0.747&0.746&0.746&0.737\\
				&$\Gamma_3$&0.986&0.929&0.001&0.000&0.328&0.328&0.328&0.328\\
				&$\Gamma_4$&0.984&0.987&0.981&0.985&0.327&0.327&0.327&0.327\\
				\hline
				\multirow{4}{*}{$\tau=3$}
				&$\Gamma_1$&0.889&0.890&0.909&0.904&1.000&1.000&1.000&1.000\\		 
				&$\Gamma_1$&0.889&0.888&0.905&0.896&0.717&0.797&0.829&0.850\\
				&$\Gamma_2$&0.889&0.888&0.905&0.896&0.746&0.755&0.747&0.746\\
				&$\Gamma_3$&0.987&0.981&0.931&0.522&0.446&0.446&0.446&0.446\\
				&$\Gamma_4$&0.985&0.987&0.992&0.986&0.327&0.327&0.327&0.327\\   
				\hline    
			\end{tabular}}
			\label{sim_unif}
		\end{table}		
From now on, the simulations reflect the case where the prior distribution itself deviates from the true distribution: $\mathcal{N}^\Pi\ne \mathcal{F}^*$. First, Table \ref{sim_unif} presents the simulation result when the true underlying distribution is a uniform distribution. While the classical $z$-based MIE always shows a consistent result since it has nothing do to with the prior distribution, the other four MIEs show better performances than the normal prior case. In particular, $\Gamma_1$ and $\Gamma_2$ provide greater reductions on the global contents compared to the previous normal true distribution case; still, the global level requirements are well satisfied, showing the global coverage probabilities are greater than $1-q=0.9$. Compared to this, $\Gamma_3$ and $\Gamma_4$ achieve larger reductions on the global expected content with the higher global coverage probabilities, implying the MIEs can reach the same reductions with less efforts. Still, $\Gamma_3$ are very sensitive to the location- or scale-wise deviations of the hyper-parameters. 
		
The shapes of the logistic and exponential distributions are quite far from the normal prior distribution. Thus, we would expect worse performances of the MIEs compared to the previous cases. For this reason, we reset the true standard deviation of the generating distribution to be $\tau=1$ to have better comparisons among the MIEs.  

		\begin{table}[ht]
			\centering
			\caption{Prior Misspecification when the True Distribution is Logistic}
			\small{
				\begin{tabular}{|c|c|c:c:c:c|c:c:c:c|}
					\hline
					\multicolumn{2}{|c|}{logistic}    & \multicolumn{4}{c|}{Global Coverage Rate} & \multicolumn{4}{c|}{Global Content}  \\
					\hline
					$1-q=0.9$ &\scriptsize{MIEs} & \scriptsize{$\eta^*=0$}           & \scriptsize{$\eta=2$}          & \scriptsize{$\eta=4$}          & \scriptsize{$\eta=6$}          & \scriptsize{$\eta^*=0$}           & \scriptsize{$\eta=2$}          & \scriptsize{$\eta=4$}          & \scriptsize{$\eta=6$}                \\
					\hline
					\multirow{4}{*}{$\tau=0.5$}   
					&$\Gamma_0$&0.897&0.920&0.895&0.894&1.000&1.000&1.000&1.000\\
					&$\Gamma_1$&0.877&0.637&0.000&0.000&0.792&0.779&0.505&0.504\\
					&$\Gamma_2$&0.896&0.917&0.895&0.902&0.863&0.840&0.905&0.851\\
					&$\Gamma_3$&0.000&0.000&0.000&0.000&0.095&0.095&0.095&0.095\\
					&$\Gamma_4$&0.267&0.244&0.246&0.231&0.179&0.179&0.179&0.179\\		
					\hline
					\multirow{4}{*}{$\tau^*=1$} 
					&$\Gamma_1$&0.892&0.904&0.905&0.910&1.000&1.000&1.000&1.000\\		
					&$\Gamma_1$&0.887&0.904&0.904&0.869&0.857&0.873&0.891&0.868\\
					&$\Gamma_2$&0.887&0.904&0.905&0.907&0.846&0.845&0.842&0.872\\
					&$\Gamma_3$&0.241&0.000&0.000&0.000&0.181&0.181&0.181&0.181\\
					&$\Gamma_4$&0.236&0.246&0.255&0.271&0.179&0.179&0.179&0.179\\
					\hline
					\multirow{4}{*}{$\tau=2$}
					&$\Gamma_1$&0.909&0.896&0.897&0.896&1.000&1.000&1.000&1.000\\		   
					&$\Gamma_1$&0.914&0.890&0.897&0.901&0.848&0.873&0.902&0.912\\
					&$\Gamma_2$&0.914&0.890&0.897&0.901&0.830&0.837&0.868&0.861\\
					&$\Gamma_3$&0.990&0.961&0.401&0.000&0.328&0.328&0.328&0.328\\
					&$\Gamma_4$&0.253&0.275&0.263&0.229&0.179&0.181&0.179&0.180\\   
					\hline    
				\end{tabular}}
				\label{sim_logis}
			\end{table}
			
Table \ref{sim_logis} provides the simulation results when the true underlying distribution is a logistic distribution. Note that the global coverage probabilities of Bayes Credible MIEs, $\Gamma_3$ and $\Gamma_4$, are lower than the nominal global level in any hyper-parameter combination except for the $\Gamma_3$ with the relatively matched prior mean ($\eta=0,2$) and diffused prior standard deviation ($\tau=2$). This is because the center value of the credible interval, the posterior mean, is off the target as it is derived based on the normal-normal model assumption. Note that the BMIE Thres with estimated hyper-parameters, $\Gamma_2$, still works well under this logistic true distribution. This implies that the hyper-parameter estimation procedure itself is still viable. Lastly, $\Gamma_1$ also shows quite reasonable results except for some cases with the concentrated prior. The results provide an example which reinforces the notion that the Bayesian approach can be totally off when prior distribution is misspecified.
			
	\begin{table}[ht]
		\centering
		\caption{Prior Misspecification when the True Distribution is Exponential}
		\small{
			\begin{tabular}{|c|c|c:c:c:c|c:c:c:c|}
				\hline
				\multicolumn{2}{|c|}{exponential}    & \multicolumn{4}{c|}{Global Coverage Rate} & \multicolumn{4}{c|}{Global Content}  \\
				\hline
				$1-q=0.9$ &\scriptsize{MIEs} & \scriptsize{$\eta^*=0$}           & \scriptsize{$\eta=2$}          & \scriptsize{$\eta=4$}          & \scriptsize{$\eta=6$}          & \scriptsize{$\eta^*=0$}           & \scriptsize{$\eta=2$}          & \scriptsize{$\eta=4$}          & \scriptsize{$\eta=6$}                \\
				\hline
				\multirow{4}{*}{$\tau=0.5$}
				&$\Gamma_0$&0.889&0.901&0.888&0.874&1.000&1.000&1.000&1.000\\   
				&$\Gamma_1$&0.675&0.877&0.883&0.000&0.792&0.779&0.741&0.504\\
				&$\Gamma_2$&0.896&0.900&0.888&0.872&0.916&0.919&0.922&0.922\\
				&$\Gamma_3$&0.000&0.000&0.000&0.000&0.095&0.095&0.095&0.095\\
				&$\Gamma_4$&0.010&0.015&0.011&0.013&0.177&0.178&0.178&0.178\\
				\hline
				\multirow{4}{*}{$\tau^*=1$} 
				&$\Gamma_0$&0.896&0.917&0.896&0.896&1.000&1.000&1.000&1.000\\
				&$\Gamma_1$&0.896&0.917&0.896&0.896&0.928&0.900&0.795&0.769\\
				&$\Gamma_2$&0.895&0.917&0.896&0.895&0.920&0.922&0.921&0.921\\
				&$\Gamma_3$&0.003&0.341&0.000&0.000&0.181&0.181&0.181&0.181\\
				&$\Gamma_4$&0.013&0.008&0.022&0.022&0.178&0.178&0.178&0.178\\
				\hline
				\multirow{4}{*}{$\tau=2$} 
				&$\Gamma_0$&0.889&0.897&0.873&0.904&1.000&1.000&1.000&1.000\\		  
				&$\Gamma_1$&0.889&0.892&0.873&0.904&0.921&0.850&0.820&0.770\\
				&$\Gamma_2$&0.889&0.890&0.873&0.904&0.905&0.914&0.913&0.918\\
				&$\Gamma_3$&0.916&0.978&0.864&0.001&0.328&0.328&0.328&0.328\\
				&$\Gamma_4$&0.018&0.006&0.015&0.014&0.178&0.178&0.177&0.177\\
				\hline     
			\end{tabular}}
			\label{sim_exp}
		\end{table}
Lastly, Table \ref{sim_exp} shows the simulation result when the true underlying distribution is exponential. It shows quite unexpected behaviors; that is, the highest global coverage probabilities are achieved not with $\eta=0$ but with $\eta=2$, which is a slightly deviated value from the truth. This would be due to the right-skewness of the true generating distribution. Still, the general behaviors are similar to those of the logistic case. While $\Gamma_3$ and $\Gamma_4$ cannot maintain their coverage probabilities in general, $\Gamma_1$ and $\Gamma_2$ show a consistent performance except for the $\Gamma_1$ with concentrated prior distribution.  
				
In conclusion, $\Gamma_3$ cannot be beaten when the prior is correctly specified in terms of the distribution as well as the corresponding hyper-parameters. It provides significant reductions on the global expected contents, while maintaining the global coverage probability at least the global level. However, when the hyper-parameters deviate from the truth, it becomes hard to satisfy the global level requirement; still, $\Gamma_4$ is robust for these misspecified hyper-parameters. Now, when the prior distribution itself deviates from the true generating distribution, then both $\Gamma_3$ and $\Gamma_4$ suffer from the low global coverage probability. When all these happen, $\Gamma_1$ and $\Gamma_2$ generally satisfy the global level requirement except for some extreme cases, providing reasonable reductions on the global expected contents. Therefore, we can conclude that BMIE Thres possesses robustness against the prior misspecification.
 
\section{Discussion and Future Works}
When a confidence interval is introduced, one faces the temptation to remove one side of the interval to reach a shorter expected length which implies a better precision of the interval estimator. Our procedure was motivated by this intuitive idea and it is designed to formalize the removal process of a side of the interval. However, the procedure relies on additional prior information. Because of this, the individual estimator is no longer within the class of classical confidence intervals but within the wider class of BIE which can be defined and evaluated by the integrals of the performance measures with respect to the prior distribution as in (\ref{eq17}). We call our procedure a Bayes MIE Thres as the integration process resembles the derivation of the Bayes risk in a general statistical decision problem.

While we utilize prior information, it is important to realize that the prior distribution, which we assume for the modeling, can be different from the true underlying distribution, which actually generates the target parameters. The simulation setting in the previous section takes this potential disagreement into consideration, generating the location parameters, $\mu_m$'s, from the distinct true distributions. Coupled with the comparison to the Bayes Credible MIEs, the results of the simulation study send a warning signal regarding the use of unjustified prior information, i.e., the issue of prior misspecification. 
The results suggest that evaluating a model only with a  correctly specified prior distribution is insufficient if the model involves prior information. That is, the model can fail drastically when the prior is misspecified. It is also evident that this agrees with one of the major criticisms of the Bayesian procedure. Consequently, we cannot simply enjoy having prior information, and the models should be used with discretion when they rely on prior information.

Under the independence assumption, the resulting BMIE Thres can be considered as a hyper-rectangular region estimator in the $M$ dimensional parameter space. In general, this would create a larger volume compared to an ellipsoidal region estimator, even after we discard one side of some intervals. However, in many situations, researchers also want to inspect individual IEs which cannot be tracked from an ellipsoidal region estimator. In this context, the independent setting for the BMIE Thres still has its own advantages despite being restrictive. In fact, the real data applications provide meaningful performances, i.e., smaller global contents, compared to the traditional $z$ and $t$-based MIE.  

One limitation of this study is the use of the FWCR for the global coverage probability. While major studies of MTPs have been reorganized based on FDR for the global type-I error rate, there exists a need for a refined global coverage probability to replace the FWCR. This is not a simple task, as the exact dual procedure of the FDR is impossible to obtain due to the non-existence of alternative hypothesis information in the MIE problem. If a refined global coverage probability is established, we can reformulate the BMIE Thres by replacing the global risk function in the optimization procedure.

We utilized a plug-in procedure in the data application to construct the BMIE Thres under unknown standard deviations in section 5.2. However, we could also have built the procedure upwards from the full prior structure for the $m$th individual IE is as follows:  
\begin{align}\label{eq71}
	\overline{X}_m|\mu_m,\lambda_m\sim& N\left(\mu_m,(n_m\lambda_m)^{-1}\right)\;\;\&\;\;&(n_m-1)S_m^2|\lambda_m\sim& g_1=Gamma(\tfrac{n_m-1}{2},\tfrac{\lambda_m}{2}) \nonumber \\
	\mu_m|\lambda_m\sim& N\left(\eta,(\kappa_m\lambda_m)^{-1} \right)\;\;\;\;\;\&\;\;&\lambda_m\sim& g_2=Gamma(a,b)
\end{align}
where the normal distributions consist of location and precision parameters and the gamma distributions consist of shape and rate parameters. Given this structure, the form of the $m$th individual BIE Thres has the following form:
\begin{align}
	\Gamma_m(X_m;\mu_m,\alpha_m)=&\left(\overline{X}_m-t^{n_m-1}_{\alpha_m/2}\tfrac{S_m}{\sqrt{n_m}}I\left\{\overline{X}_m-\eta>-C\sqrt{\tfrac{b}{\kappa_m}}\tfrac{\Gamma(a-\tfrac{1}{2})}{\Gamma(a)}\right\},\right. \nonumber  \\  &\hspace{6pt}\left.\overline{X}_m+t^{n_m-1}_{\alpha_m/2}\tfrac{S_m}{\sqrt{n_m}}I\left\{\overline{X}_m-\eta<C\sqrt{\tfrac{b}{\kappa_m}}\tfrac{\Gamma(a-\tfrac{1}{2})}{\Gamma(a)}  \right\} \right)
\end{align}
where $S_m^2=\tfrac{1}{n_m-1}\sum_{i=1}^{n_m}(X_i-\overline{X}_m)^2$ and $t^{n_m-1}_{\alpha_m/2}$ is the $1-\alpha_m/2$ quantile of a $t$-distribution with $n_m-1$ degrees of freedom.
The corresponding BCP and BEL can be derived, but no closed forms exist. The hardest part is the estimation of the hyper parameters: $\eta$, $a$, $b$, and $\kappa_m$'s. If this issue can be resolved, then we will be able to apply the resulting BMIE Thres to the actual data for unknown standard deviations. 

Another interesting direction for future study of BMIE Thres is its nonparametric extensions. Note that under the normal-normal model, BMIE Thres shares the same setting as the parametric empirical Bayes. Therefore, a reasonable starting point for the extension would be the nonparametric empirical Bayes setting, and we could derive a threshold based on that nonparametric structure. Lastly, a thresholding approach which does not rely on a prior distribution can have practical utility. In this situation, the thresholding decision can be done through the information from a domain study such as economics, engineering, or biology.

\section*{Acknowlegments}
This manuscript is a part of the first author's dissertation work at the University of South Carolina. 
The authors thank Professor James Lynch, Dr. Shiwen Shen, and graduate students Jeff Thompson, Lu Wang, and Lili Tong for helpful discussions during research seminars.
In addition, we thank Professors Karl Gregory, Dewei Wang and George Androulakis for their helpful comments as dissertation committee members of the first author. We also acknowledge NIH grant P30GM103336-01A1 and the Center for Colon Cancer Research at the University of South Carolina for partially supporting this research. Lastly, the first author appreciate Professor Alexander Goldenshluger for providing the post-doctoral research opportunity at the University of Haifa in Israel. 

\newpage
\bibliographystyle{plainnat}
\bibliography{foo}

\begin{thebibliography}{26}
\providecommand{\natexlab}[1]{#1}
\providecommand{\url}[1]{\texttt{#1}}
\expandafter\ifx\csname urlstyle\endcsname\relax
  \providecommand{\doi}[1]{doi: #1}\else
  \providecommand{\doi}{doi: \begingroup \urlstyle{rm}\Url}\fi

\bibitem[Benjamini and Hochberg(1995)]{benjamini:1995}
Yoav Benjamini and Yosef Hochberg.
\newblock Controlling the false discovery rate: a practical and powerful
  approach to multiple testing.
\newblock \emph{Journal of the royal statistical society. Series B
  (Methodological)}, pages 289--300, 1995.

\bibitem[Benjamini and Yekutieli(2005)]{benjamini:2005}
Yoav Benjamini and Daniel Yekutieli.
\newblock False discovery rate--adjusted multiple confidence intervals for
  selected parameters.
\newblock \emph{Journal of the American Statistical Association}, 100\penalty0
  (469):\penalty0 71--81, 2005.

\bibitem[Berger(2013)]{berger:2013}
James~O Berger.
\newblock \emph{Statistical decision theory and Bayesian analysis}.
\newblock Springer Science \& Business Media, 2013.

\bibitem[Brown(2008)]{brown:2008}
Lawrence~D Brown.
\newblock In-season prediction of batting averages: A field test of empirical
  bayes and bayes methodologies.
\newblock \emph{The Annals of Applied Statistics}, pages 113--152, 2008.

\bibitem[Casella and Hwang(1983)]{casella:1983}
George Casella and Jiunn~Tzon Hwang.
\newblock Empirical bayes confidence sets for the mean of a multivariate normal
  distribution.
\newblock \emph{Journal of the American Statistical Association}, 78\penalty0
  (383):\penalty0 688--698, 1983.

\bibitem[Casella and Hwang(1991)]{casella:1991}
George Casella and Jiunn~Tzon Hwang.
\newblock Evaluating confidence sets using loss functions.
\newblock \emph{Statistica Sinica}, pages 159--173, 1991.

\bibitem[Dobriban et~al.(2015)Dobriban, Fortney, Kim, and Owen]{dobriban:2015}
Edgar Dobriban, Kristen Fortney, Stuart~K Kim, and Art~B Owen.
\newblock Optimal multiple testing under a gaussian prior on the effect sizes.
\newblock \emph{Biometrika}, 102\penalty0 (4):\penalty0 753--766, 2015.

\bibitem[Efron(2012)]{efron:2012}
Bradley Efron.
\newblock \emph{Large-scale inference: empirical Bayes methods for estimation,
  testing, and prediction}, volume~1.
\newblock Cambridge University Press, 2012.

\bibitem[Efron and Hastie(2016)]{efron:2016}
Bradley Efron and Trevor Hastie.
\newblock \emph{Computer age statistical inference}, volume~5.
\newblock Cambridge University Press, 2016.

\bibitem[Efron and Morris(1975)]{efron:1975}
Bradley Efron and Carl Morris.
\newblock Data analysis using stein's estimator and its generalizations.
\newblock \emph{Journal of the American Statistical Association}, 70\penalty0
  (350):\penalty0 311--319, 1975.

\bibitem[Efron and Morris(1977)]{efron:1977}
Bradley Efron and Carl Morris.
\newblock Stein's paradox in statistics.
\newblock \emph{Scientific American}, 236\penalty0 (5):\penalty0 119--127,
  1977.

\bibitem[Fithian et~al.(2014)Fithian, Sun, and Taylor]{fithian:2014}
William Fithian, Dennis Sun, and Jonathan Taylor.
\newblock Optimal inference after model selection.
\newblock \emph{arXiv preprint arXiv:1410.2597}, 2014.

\bibitem[Good(1983)]{good:1983}
Irving~John Good.
\newblock \emph{Good thinking: The foundations of probability and its
  applications}.
\newblock U of Minnesota Press, 1983.

\bibitem[Habiger and Pe{\~n}a(2014)]{habiger:2014}
Joshua~D Habiger and Edsel~A Pe{\~n}a.
\newblock Compound p-value statistics for multiple testing procedures.
\newblock \emph{Journal of multivariate analysis}, 126:\penalty0 153--166,
  2014.

\bibitem[Hochberg and Tamhane(1987)]{hochberg:1987}
Y.~Hochberg and A.~C. Tamhane.
\newblock \emph{Multiple Comparison Procedures}.
\newblock New York: Wiley, 1987.

\bibitem[Hochberg(1988)]{hochberg:1988}
Yosef Hochberg.
\newblock A sharper bonferroni procedure for multiple tests of significance.
\newblock \emph{Biometrika}, 75\penalty0 (4):\penalty0 800--802, 1988.

\bibitem[Holm(1979)]{holm:1979}
Sture Holm.
\newblock A simple sequentially rejective multiple test procedure.
\newblock \emph{Scandinavian journal of statistics}, pages 65--70, 1979.

\bibitem[Lehmann and Romano(2006)]{lehmann:2006}
Erich~L Lehmann and Joseph~P Romano.
\newblock \emph{Testing statistical hypotheses}.
\newblock Springer Science \& Business Media, 2006.

\bibitem[Miller(1966)]{miller:1966}
Rupert~G. Miller.
\newblock \emph{Simultaneous statistical inference}.
\newblock McGraw-Hill, 1966.

\bibitem[Morris(1983)]{morris:1983}
Carl~N Morris.
\newblock Parametric empirical bayes inference: theory and applications.
\newblock \emph{Journal of the American Statistical Association}, 78\penalty0
  (381):\penalty0 47--55, 1983.

\bibitem[Pe{\~n}a et~al.(2011)Pe{\~n}a, Habiger, and Wu]{pena:2011}
Edsel~A Pe{\~n}a, Joshua~D Habiger, and Wensong Wu.
\newblock Power-enhanced multiple decision functions controlling family-wise
  error and false discovery rates.
\newblock \emph{Annals of statistics}, 39\penalty0 (1):\penalty0 556, 2011.

\bibitem[Roy and Bose(1953)]{roy:1953}
SN~Roy and Raj~Chandra Bose.
\newblock Simultaneous confidence interval estimation.
\newblock \emph{The Annals of Mathematical Statistics}, pages 513--536, 1953.

\bibitem[Scheff{\'e}(1953)]{scheffe:1953}
Henry Scheff{\'e}.
\newblock A method for judging all contrasts in the analysis of variance.
\newblock \emph{Biometrika}, 40\penalty0 (1-2):\penalty0 87--110, 1953.

\bibitem[Shaffer(1995)]{shaffer:1995}
Juliet~Popper Shaffer.
\newblock Multiple hypothesis testing.
\newblock \emph{Annual review of psychology}, 46\penalty0 (1):\penalty0
  561--584, 1995.

\bibitem[Westfall and Young(1993)]{westfall:1993}
P~Westfall and S.~Young.
\newblock \emph{Resampling-Based Multiple Testing}.
\newblock New York: Wiley, 1993.

\bibitem[Westfall et~al.(1998)Westfall, Krishen, and Young]{westfall:1998}
Peter~H Westfall, Alok Krishen, and S~Stanley Young.
\newblock Using prior information to allocate significance levels for multiple
  endpoints.
\newblock \emph{Statistics in medicine}, 17\penalty0 (18):\penalty0 2107--2119,
  1998.

\end{thebibliography}

\newpage
\appendix
\addtocontents{toc}{\setcounter{tocdepth}{1}}

\section{Optimization: Decision Theoretic Approach}
This appendix introduces an optimization based on a decision theoretic framework through the two global risk functions which reflect the global coverage probability and global expected content. The optimization procedure searches for the best MIE by allocating the optimal individual levels so that the global expected content can be minimized while maintaining the global coverage probability at a global level of at least $1-q$. The general idea is adopted from \citet{pena:2011} which aims at the best MTP by allocating the optimal individual sizes under the family-wise error rate or false discovery rate. The allocation procedure is called \textit{size investing strategy}. Similarly, the goal of the procedure in this appendix is to establish a \textit{confidence level investing strategy}.

\subsection{Individual Loss}
\leavevmode
Let $\Theta=(-\infty,\infty)$ be a parameter space and $\mathcal{A}=\{(a_1,a_2):-\infty<a_1<a_2<\infty\}$ be an action space. 
Now, we define a pair of loss functions $L_0$ an $L_1$ as follows:
\begin{align}
L_0(\theta,w)=\nu(w)\;\&\;\;L_1(\theta,w)=I\{\theta\notin w\}=I\{\theta\notin(a_1,a_2)\}
\end{align}
where $w=(a_1,a_2)$ and $\nu$ is the content measure. In this study, $\nu(w)=|a_2-a_1|$ because $\theta$ is a location parameter. Note that the first loss function penalizes intervals with wide lengths and the second loss function penalizes intervals which do not contain true parameters. Given a small positive number $\alpha\in(0,1)$, we set up an optimization problem as follows: for every $\theta\in\Theta$,
\begin{align}
\text{minimize }E_\theta[L_0(\theta,W(X))] \text{ subject to }E_\theta[L_1(\theta,W(X))]\leq \alpha.    
\end{align}
Note that this setting represents the usual pursuit of the tightest IE with the coverage probability maintained at least at a nominal level of $1-\alpha$.

One issue with the loss functions is that, whereas the range of $L_1(\theta,w)$ is always between zero and one, the range of $L_0(\theta,w)$ is positively unbounded. This lack of balance between two loss functions can cause an unstable result in the optimization. To handle this issue, we adjust the $L_0(\theta,w)$ by adopting a function, $h_\beta(x)=\tfrac{x}{\beta+x}$, where $\beta$ is a positive constant as follows:
\begin{align}
L_0^\beta(\theta,w)=h_\beta(L_0(\theta,w))=\tfrac{\nu(w)}{\beta+\nu(w)}.
\end{align}
Note that the adjusted loss function $L_0^\beta(\theta,w)$ ranges from zero to one on the positive domain. Moreover, ths function, $h_\beta(x)$, is smooth with the $n$th derivative, $\tfrac{d^n}{dx^n}h_\beta(x)=\tfrac{(-1)^{n+1}n!\beta}{(\beta+x)^{n+1}}$, which are utilized to the optimization procedure. A similar loss function approach was introduced in \citet{casella:1991} for a single-dimensional case. We extend our idea to a multi-dimensional case in the next section.

\subsection{Global Loss and Risk}
\leavevmode
Let $M$ be a positive integer. Then $\bc{\Theta}=(-\infty,\infty)^M$ is a paramter space with an element $\boldsymbol{\theta}=(\theta_1,\theta_2,\ldots,\theta_M)^T$, and $\boldsymbol{\mathcal{A}}=\{\times_{m=1}^M(a^m_1,a^m_2):-\infty<a^m_1<a^m_2<\infty,\;m=1,$$\ldots,M\}$ is an action space with an element $(\boldsymbol{a_1,a_2})=[(a^1_1,a^1_2), (a^2_1,a^2_2),$ $ \ldots,  (a^M_1,a^M_2)]^T$. Now we define two global loss functions as follows:
\begin{align}
	\boldsymbol{L}^\beta_0(\boldsymbol{\theta},\boldsymbol{w})=\frac{1}{M}\sum_{m=1}^ML_0^\beta(\theta_m,w_m)\;\&\;\;
	\boldsymbol{L}_1(\boldsymbol{\theta},\boldsymbol{w})=I\left\{\left(\sum_{m=1}^ML_1(\theta_m,w_m)\right)\geq 1\right\} 
\end{align}
where $\boldsymbol{w}=(\boldsymbol{a_1,a_2})$. Observe that the interpretations of the individual loss functions are still maintained in the global loss functions. That is, the first global loss function penalizes multiple intervals with wide contents, and the second global loss function penalizes multiple intervals which do not cover at least one true parameter.

Now, let $\mathcal{D}$ be a class of nonrandomized multiple decision functions which consist of $\boldsymbol{\delta:\mathcal{X}\longrightarrow\mathcal{A}}$ where $\boldsymbol{\delta(X;\alpha)}=[\delta_1(X;\alpha)=(LB_1(X;\alpha),UB_1(X;\alpha)),\ldots,\delta_M(X;\alpha)=$\\$(LB_M(X;\alpha),UB_M(X;\alpha))]^{T}$. In order to obtain the risk functions, we need to take expectations on the global loss functions. However, $\boldsymbol{L}^\beta_0(\boldsymbol{\theta},\boldsymbol{w})$ has a non-linear form with respect to the random variable, so we cannot take that expectation directly. However, this issue can be circumvented by using a linear interpolation, i.e., the risk function can be well approximated. The resulting risk functions are as follows: 
\begin{align}
\boldsymbol{R_0^\beta(\theta,\delta)}=\boldsymbol{E_\theta\left[L^\beta_0(\theta,\delta(X;\alpha))\right]}\;\&\;\boldsymbol{R_1(\theta,\delta)}=\boldsymbol{E_\theta\left[L_1(\theta,\delta(X;\alpha))\right].}
\end{align}
Notice that the first global risk function is the adjusted global expected content. In addition, the second global risk function is related to the familywise coverage rate (FWCR) which is defined as the probability that an MIE covers all of the true parameters, so that 
$\boldsymbol{R_1(\theta,\delta)=1-FWCR(\theta,\delta)}.$

\subsection{Optimization Procedure}
\leavevmode
Given a small positive number $q\in(0,1)$, we set up an optimization problem as follows:
\begin{align}
\text{minimize}\;\; \boldsymbol{R^\beta_0(\theta,\delta)}\;\;\text{subject to}\;\;\boldsymbol{R_1(\theta,\delta)}\leq q.
\end{align}
Note that the restriction implies the global coverate probability, $\bc{FWCR(\theta,\delta)}$, is maintained to be at least a global level, $1-q$.
In this subsection, we apply this procedure to an MIE for $M$ normal location parameters with known variances. By the Sufficiency Principle, we simplify the setting as follows:
\[
\bar{X}_m\sim N(\mu_m,\sigma_m^2)\;\;\text{for}\;\; m=1,\ldots,M,
\]
where the random variables are independent throughout the index $m$ and the variance of $\bar{X}_m$ is set to be $\sigma^2_m$ without loss of generality. The form of the $m$th individual IE, $\Gamma_m$, has the usual form as follows:
\[
\Gamma_m(X_m;\alpha_m)=[LB_m(X_m;\alpha_m),UB_m(X_m;\alpha_m)]=[\bar{X}_m-z_{\alpha_m/2}\sigma_m,\bar{X}_m+z_{\alpha_m/2}\sigma_m]
\]
where $z_{\alpha}=\Phi^{-1}(1-\alpha)$. As mentioned earlier, the content measure $\nu$ is the Lebesgue measure, $v(w)=|a_2-a_1|$, because the mean is a location parameter. Given this setting, we evaluate the two risk functions as follows:
\begin{align}
	\boldsymbol{R^\beta_0(\mu,\Gamma)}\;=\;&\bc{E_{\mu}}\left[\frac{1}{M}\sum_{m=1}^Mh_\beta(\nu(LB_m(X_m;\alpha_m),UB_m(X_m;\alpha_m)))\right]\nonumber\\
	\approx\;&\frac{1}{M}\sum_{m=1}^Mh_\beta\left(E_{\mu_m}[UB_m(X_m;\alpha_m)-LB_m(X_m;\alpha_m)]\right)\nonumber\\
	=\;&\frac{1}{M}\sum_{m=1}^Mh_\beta(2z_{\alpha_m/2}\sigma_m)\nonumber\\
	=\;&\frac{1}{M}\sum_{m=1}^M\frac{2z_{\alpha_m/2}\sigma_m}{\beta+2z_{\alpha_m/2}\sigma_m}\\
	\boldsymbol{R_1(\mu,\Gamma)}\;=\;&1-\boldsymbol{P_\mu}\left[\left(\sum_{m=1}^ML_1(\mu_m,w_m)\right)=0\right]\nonumber\\
	=\;&1-\boldsymbol{P_\mu}\left[ \bigcap_{m=1}^M \{\mu_m\in (LB_m(X_m;\alpha_m),UB_m(X_m;\alpha_m))\}\right]\nonumber\\
	=\;&1-\prod_{m=1}^MP_{\mu_m}\left[ \mu_m\in (LB_m(X_m;\alpha_m),UB_m(X_m;\alpha_m))  \right]\nonumber\\
	=\;&1-\prod_{m=1}^M (1-\alpha_m)
\end{align}
Note the approximation in the second equality can be achieved by a piece-wise linear interpolation. We reparametrize $\nu_m=z_{\alpha_m/2}$ for a numerical stability. Then the initial optimization problem can be restated as follows:   
\begin{align}
\text{minimize}\;\; \frac{1}{M}\sum_{m=1}^M\frac{2\nu_m\sigma_m}{\beta+2\nu_m\sigma_m}\;\;
\text{subject to}\;\;\sum_{m=1}^M\log(2\Phi(\nu_m)-1)\geq \log(1-q).
\end{align}
This problem can be numerically solved by using the Newton-Raphson method after setting up a Lagrange equation. 

\subsection{Optimization Result}
\leavevmode
The essence of the optimization procedure is the allocation of optimal levels to the individual IEs, called the confidence level investing strategy.
This process allows us to reach the smallest global expected content, while maintaining the global level requirement. In this problem, the tuning parameter, $\beta$, determines the shape of the function $h_\beta$ which controls the allocation strategy.
\begin{figure}[ht]
	\centering
	\includegraphics[scale=0.45]{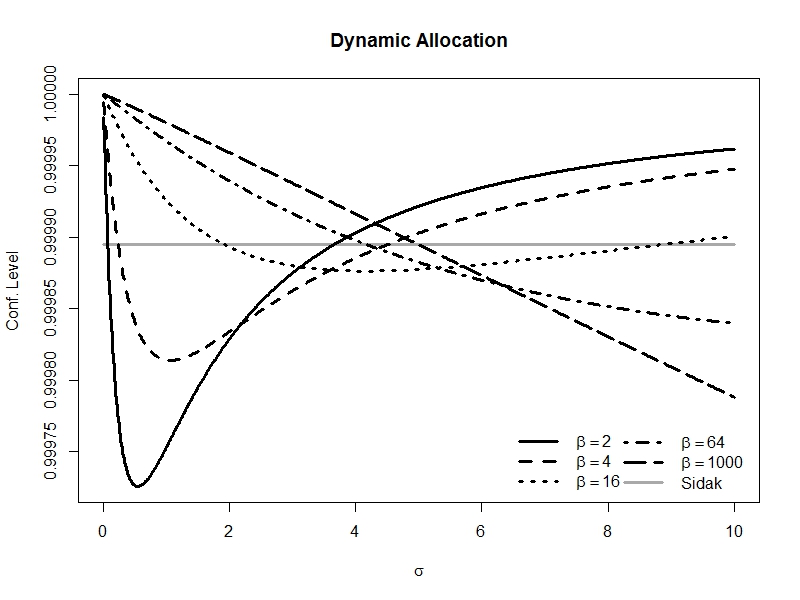}
	\caption[Individual Level Allocation: Normal Mean]{Individual Level Allocation: Normal Mean} 
	\label{CLA} 
\end{figure}
Figure \ref{CLA} illustrates the dynamics due to $\beta$ in the optimal level allocations of the MIE for $M=1000$ normal location parameters, $\mu_m$'s.
Notice that an equi-spaced sequence from 0.01 to 10 is assigned for $\sigma_m$'s. Therefore, the points on the graphs represent the allocated individual levels with respect to the $\sigma_m$'s on the horizontal line. First, the gray horizontal line represents the allocation based on the Sidak adjustment, i.e., the constant individual levels, $(1-q)^{1/M}$. Compared to this, the black curves represent the allocation results achieved through the optimization procedure. The shape of the curves varies with respect to the value of $\beta$. When $\beta$ is small, the trend shows a highly nonlinear shape, assigning large individual levels to the IEs with very small and large $\sigma$'s. However, as $\beta$ becomes larger, the form of the curves reaches an almost linear line with the negative slope. Eventually, when $\beta$ is greater than 1000, the shape of the curve remains invariant with respect to the value of the $\beta$.
In this application, the global expected content defined as a relative expected length (REL), i.e., the average expected length ratio of the MIE with optimal levels to the MIE with the Sidak adjustment. 

Table \ref{RLLT} summarizes the global RELs with respect to $\beta$.
\begin{table}[ht]
	\centering
	\caption[Relative Expected Length: Normal Mean]{Relative Expected Length with respect to $\beta$}
	\begin{tabular}{|c|cccccc|}
		\hline
		$\beta$& 1 & 2 & 8 & 32 & 1000 & Sidak Adj. \\
		\hline
		REL    & 1.0317 & 1.0177 & 0.9967 & 0.9888 & 0.9874 & 1 \\
		\hline
	\end{tabular}
	\label{RLLT}
\end{table}
Note that when $\beta$ is small, the REL is greater than 1, implying the performance is no better than the MIE with the Sidak adjustment, the constant allocation. However, the REL becomes smaller and converges to 0.9874 as $\beta$ increases.  
This overall reduction can be explained through the allocation result. When $\beta=1000$, the optimization procedure assigns smaller individual levels to the IEs with larger $\sigma_m$'s to counterbalance the sizes of $\sigma_m$'s with smaller $z_{\alpha_m/2}$'s. To compensate for these \textit{investments}, the procedure matches larger $z_{\alpha_m/2}$'s to the smaller $\sigma_m$'s by assigning larger levels to the corresponding IEs. These processes are performed simultaneously to minimize the REL, maintaining the FWCR at least the global level of $1-q=0.9$. However, the resulting 1.26\% reduction in the relative expected length is quite limited amount. This limitation particularly motivates the thresholding approach in this study. 

\section{Proofs of Propositions and Lemma}
\propa*
\begin{proof}
	We can first manipulate the original length as follows.
	\begin{align*}
	\nu(\Gamma_m(\bc{X}_m;\mu_m,\alpha_m))=&z_{\alpha_m/2}\sigma_m\left[I\left\{\bar{X}_m<\eta+C\tau\right\}+I\left\{\bar{X}_m>\eta-C\tau\right\}\right]\\
	=&z_{\alpha_m/2}\sigma_m\left[1+I\left\{\eta-C\tau<\bar{X}_m<\eta+C\tau\right\} \right]
	\end{align*}
	where $\nu$ is the measure of the content. It is the Leabesgue measure in this case of location parameters. To derive the Bayes length, we take the expectations as follows: 
	\begin{align*}
	&\int_{\Theta_m}\int_{\mathcal{X}_m}\nu(\Gamma_m(x_m;\mu_m,\alpha_m))dP_m(x_m)d\Pi(\mu_m)\\
	=&E_{\mu_m}\left[E_{\bar{X}_m|\mu_m}[\nu(\Gamma_m(\bc{X}_m;\mu_m,\alpha_m))]\right]\\
	=&z_{\alpha_m/2}\sigma_m\left(1+E_{\mu_m}P_{Z_m}\left[\tfrac{\eta-C\tau-\mu_m}{\sigma_m}<Z_m<\tfrac{\eta+C\tau-\mu_m}{\sigma_m}\right] \right)\\
	=&z_{\alpha_m/2}\sigma_m\left(1-E_{\mu_m}\left[\Phi\left(\tfrac{\mu_m-\eta-C\tau}{\sigma_m}\right)-\Phi\left(\tfrac{\mu_m-\eta+C\tau}{\sigma_m}\right)\right]\right)
	\\=&z_{\alpha_m/2}\sigma_m\left(1-E_{Z'_m}\left[\Phi\left(\tfrac{Z'_m-C}{\sigma_m/\tau}\right)-\Phi\left(\tfrac{Z'_m+C}{\sigma_m/\tau}\right)\right]\right)\\
	=&z_{\alpha_m/2}\sigma_m\left(1-\Phi\left(\tfrac{-C\tau}{\sqrt{\sigma_m^2+
			\tau^2}}\right)+\Phi\left(\tfrac{C\tau}{\sqrt{\sigma_m^2+\tau^2}}\right) \right) 
	=2z_{\alpha_m/2}\sigma_m\Phi\left(\tfrac{C\tau}{\sqrt{\sigma_m^2+
			\tau^2}}\right)
	\end{align*}
	Once the form is derived, it is easy to observe, as $C$ goes to infinity, this Bayes length approaches the expected length of the classical $z$-based IE, $2z_{\alpha_m/2}\sigma_m$.
\end{proof}

\propb*
\begin{proof}
	For the derivation, it is better to use the posterior and marginal distributions:
	\[\mu_m|\bar{X}_m\sim N\left(\tfrac{\tau^2}{\tau^2+\sigma_m^2}\bar{X}_m+\tfrac{\sigma_m^2}{\tau^2+\sigma_m^2}\eta,\tfrac{\tau^2\sigma_m^2}{\tau^2+\sigma_m^2}\right)\;\;\&\;\;\bar{X}_m\sim N\left(\eta,\sigma_m^2+\tau^2\right)\]
	Then, the Bayes coverage probability becomes:
	\begin{align*}
	&\int_{\Theta_m} P_{\mu_m}[\mu_m\in\Gamma_m(\bc{X}_m;\mu_m,\alpha_m)]d\Pi(\mu_m)\\
	=&E_{\bar{X}_m}E_{\mu_m|\bar{X}_m}I\left\{LB_m\leq\mu_m\leq UB_m\right\}\\
	=&E_{\bar{X}_m}P_{\mu_m|\bar{X}_m}\left[\mu_m<\bar{X}_m+z_{\alpha_m/2}\sigma_mI\left\{\bar{X}_m-\eta<C\tau\right\}\right]\\
	-&E_{\bar{X}_m}P_{\mu_m|\bar{X}_m}\left[\mu_m<\bar{X}_m-z_{\alpha_m/2}\sigma_mI\left\{\bar{X}_m-\eta>-C\tau\right\}\right]\\
	=&E_{Z'_m}\Phi\left(\tfrac{\sigma_m}{\tau}Z'_m+\tfrac{\sqrt{\tau^2_m+\sigma^2_m}}{\tau}z_{\alpha_m/2}I\left\{Z'_m<\tfrac{C\tau}{\sqrt{\tau^2+\sigma_m^2}}\right\}  \right)\\
	-&E_{Z'_m}\Phi\left(\tfrac{\sigma_m}{\tau}Z'_m-\tfrac{\sqrt{\tau^2_m+\sigma^2_m}}{\tau}z_{\alpha_m/2}I\left\{Z'_m>\tfrac{-C\tau}{\sqrt{\tau^2+\sigma_m^2}}\right\}  \right)\\
	=&\int_{-C_m}^{C_m}\left\{\Phi\left(\tfrac{\sigma_m}{\tau}z'_m+\tfrac{\sqrt{\tau^2_m+\sigma^2_m}}{\tau}z_{\alpha_m/2}\right)-\Phi\left(\tfrac{\sigma_m}{\tau}z'_m-\tfrac{\sqrt{\tau^2_m+\sigma^2_m}}{\tau}z_{\alpha_m/2}\right)\right\}\phi(z'_m)dz'_m\\
	+&\int_{C_m}^\infty\left\{\Phi\left(\tfrac{\sigma_m}{\tau}z'_m\right)-\Phi\left(\tfrac{\sigma_m}{\tau}z'_m-\tfrac{\sqrt{\tau^2_m+\sigma^2_m}}{\tau}z_{\alpha_m/2}\right)\right\}\phi(z'_m)dz'_m\\
	+&\int_{-\infty}^{-C_m}\left\{\Phi\left(\tfrac{\sigma_m}{\tau}z'_m+\tfrac{\sqrt{\tau^2_m+\sigma^2_m}}{\tau}z_{\alpha_m/2}\right)-\Phi\left(\tfrac{\sigma_m}{\tau}z'_m\right)\right\}\phi(z'_m)dz'_m\\
	=&2\int_{-\infty}^{C_m}\left\{\Phi\left(\tfrac{\sigma_m}{\tau}z'_m+\tfrac{\sqrt{\tau^2_m+\sigma^2_m}}{\tau}z_{\alpha_m/2}\right)-\Phi\left(\tfrac{\sigma_m}{\tau}z'_m\right)\right\}\phi(z'_m)dz'_m
	\end{align*}
	where $C_m=\tfrac{C\tau}{\sqrt{\tau^2+\sigma_m^2}}$.
	Although it has no closed form, it is a smoothly increasing function with respect to $C$ and approaches the coverage probability of the classical $z$-based IE, $1-\alpha_m$. That is,
	\begin{align*}
	\lim_{C\to\infty}T_m(\alpha_m,C)=&2\int_{-\infty}^\infty\left\{\Phi\left(\tfrac{y+\tfrac{\sqrt{\tau^2+\sigma_m^2}}{\sigma_m}z_{\alpha_m/2}}{\tau/\sigma_m}\right)-\Phi\left(\tfrac{y-0}{\tau/\sigma_m}   \right)\right\}\phi(y)dy\\
	=&2\left\{\Phi\left(\Phi^{-1}\left(\alpha_m/2\right)\right)-\tfrac{1}{2}\right\}=1-\alpha_m
	\end{align*}\end{proof}

\propc*
\begin{proof}
	Note that $\bar{X}_m$ marginally follows $\mathcal{N}(\eta,\sigma_m^2+\tau^2)$. Then,
	\begin{align*}
	&E_{\bar{X}_m}\left[\frac{1}{M}\sum_{i=1}^M\left[I\{\bar{X}_m>\eta+C\tau\}+I\{\bar{X}_m<\eta-C\tau\} \right]  \right]\\
	=&\frac{1}{M}\sum_{i=1}^M\left[P[\bar{X}_m > \eta+C\tau]+P[ \bar{X}_m < \eta-C\tau ]\right]\\
	=&\frac{1}{M}\sum_{i=1}^M\left[P\left[\frac{\bar{X}_m-\eta}{\sqrt{\sigma_m^2+\tau^2}} > \frac{C\tau}{\sqrt{\sigma_m^2+\tau^2}}\right]+P\left[ \frac{\bar{X}_m-\eta}{\sqrt{\sigma_m^2+\tau^2}} < \frac{-C\tau}{\sqrt{\sigma_m^2+\tau^2}} \right]\right]\\
	=&\frac{1}{M}\sum_{i=1}^M\left[1-\Phi\left(\frac{C\tau}{\sqrt{\sigma_m^2+\tau^2}}\right)+\Phi\left( \frac{-C\tau}{\sqrt{\sigma_m^2+\tau^2}}  \right)  \right]\\
	=&\frac{2}{M}\sum_{i=1}^M\Phi\left( \frac{-C\tau}{\sqrt{\sigma_m^2+\tau^2}}\right)=\frac{2}{M}\sum_{i=1}^M\Phi(-C_m)
	\end{align*}
	where $C_m=C\tau/\sqrt{\sigma_m^2+\tau^2}$.
\end{proof}

\coro*
\begin{proof}
	From Proposition \ref{prop1} and Definition \ref{defin3}, we can obtain the BAELs of BMIE Thres and $z$-based MIE, respectively. Then $BREL[\bc{\mu},\bc{\alpha},C;\bc{\sigma},\tau]$ is the ratio of these two BAELs. In addition, the $BFWCR[\bc{\mu},\bc{\alpha},C;\bc{\sigma},\tau]$ can be obtained from Proposition \ref{prop2} and Definition \ref{defin2} by multiplying $M$ individual BIEs. Lastly, $BTR[\bc{\mu},\bc{\alpha},C;\bc{\sigma},\tau]$ is the immediate result from Proposition \ref{prop3}.
\end{proof}

\end{document}